\documentclass[10pt]{article}
\usepackage[letterpaper,margin=1.05in,includefoot,footskip=0.35in]{geometry}
\usepackage{amsmath,amssymb,amsthm}
\usepackage{graphicx}
\usepackage{booktabs}
\usepackage{float}
\usepackage[labelsep=period]{caption}
\usepackage{natbib}
\usepackage{microtype}
\usepackage[hidelinks]{hyperref}

\newtheorem{definition}{Definition}
\newtheorem{lemma}{Lemma}
\newtheorem{proposition}{Proposition}

\newcommand{\sem}[1]{\left[\!\left[#1\right]\!\right]}
\title{Executable Boundary Contracts for Sound Event Traces}
\author{
Faruk Alpay\\
Department of Computer Engineering\\
Bahcesehir University, Istanbul, Turkey\\
\texttt{faruk.alpay@bahcesehir.edu.tr}\\
\textit{Correspondence}\quad \texttt{alpay@lightcap.ai}
\and
Hamdi Alakkad\\
Department of Artificial Intelligence Engineering\\
Bahcesehir University, Istanbul, Turkey\\
\texttt{hamdi.alakkad@bahcesehir.edu.tr}
}
\date{}

\begin{document}
\maketitle

\begin{abstract}
Sound event reports often compress
timed boundary behavior into frame,
segment, or event scores. This paper
defines executable boundary contracts
for finite sound event traces. The
frame fragment is a bounded Boolean
fragment embeddable in STL after grid
projection. The event layer adds
declared interval matching, duration
clauses, fragmentation clauses, and
obligation restricted vector scoring.
The aim is measurement, not a new
general temporal logic and not a
challenge leaderboard. The artifact
evaluates controlled Mini LibriSpeech
seeded scenes, MAESTRO Real
soundscapes, frozen pretrained timing
probes, and an official DCASE 2024
Task 4 baseline track. Across these
tracks, standard scores and contract
coordinates disagree in interpretable
ways. The strongest real corpus
finding is that union activity can
hide typed boundary failure, while
external DCASE outputs provide a class
indexed challenge level reference.
Code, generated tables, manifests, and
Lean checks for the finite frame core
are supplied as ancillary material.

\end{abstract}

\section{Introduction}
Auditory boundary benchmarks are often
written as classifier evaluations. A
detector receives sound and emits a
label per frame or an interval list.
The report then gives frame F1,
segment F1, event F1, or a challenge
specific aggregate. Those numbers
summarize overlap, but the executable
object downstream is a trace.

The downstream object is a finite
timed trace. A speech gate opens and
closes. A sound event front end
triggers and releases. A retrieval
system receives a predicted segment. A
brain signal decoder aligns a neural
trace to a candidate audio passage. In
these cases a boundary error is not
only a loss of frame overlap. It is a
change in the trace contract offered
to the next component.

The central object is semantic and
operational. A benchmark has trace
obligations, accepted temporal
displacement, fragmentation policy,
silence protection, and matched event
comparison. A scalar metric
approximates that policy after the
fact. A formula states it before the
run.

The proposed contract calculus defines
a finite specification language for
that policy. The parser accepts atomic
predicates over reference and
prediction traces, bounded
neighborhood modalities, bounded
future modalities, bounded until,
negation, conjunction, disjunction,
and implication. It produces an
abstract syntax tree that is evaluated
over a fixed frame grid.

The language stays with finite words
and finite evidence. Its semantic
object is the trace emitted by a
detector, together with the
obligations that trace is expected to
satisfy.

The interval layer is equally
explicit. Maximal active runs become
half open intervals. A transparent
matching relation pairs reference and
prediction intervals. Duration and
fragmentation clauses are typed event
clauses over that relation. The frame
formulas and event clauses are
combined inside one finite contract
vector.

The empirical part of the paper is
built to stress the distinction. We
construct controlled auditory scenes
with Mini LibriSpeech utterance
segments, chirped tones, band limited
bursts, overlapping class events, and
unlabeled distractors. We perturb them
with noise, reverberation, clipping,
gain drift, and temporal jitter. The
paper also runs a MAESTRO Real
protocol over real TUT soundscapes
with crowdsourced soft strong labels.

Six local detectors are used. Two are
threshold systems. One is a balanced
logistic regression over frame
features. Two are clean trained
temporal convolutional networks with
different receptive fields. The sixth
detector is a boundary aware residual
temporal model trained on augmented
views. A separate frozen encoder probe
evaluates representations from
wav2vec2, wav2vec2 Conformer, AST, HTS
AT through CLAP, and BEATs under the
same contract monitor.

The benchmark reports conventional
frame and boundary scores, but it also
reports transition region F1, formula
satisfaction, tolerance sweeps, a
clean data learning curve, and class
macro scores. The transition region
score isolates frames near onsets and
offsets. The tolerance sweep shows how
rapidly conclusions change as the
accepted boundary radius changes. The
learning curve asks how the parsed
contract responds when the residual
dilated baseline receives less clean
training data.

The paper uses logic as measurement
infrastructure for sound. The object
is neither language generation nor
semantic recursion. It is a finite
auditory trace with typed boundaries
and checkable obligations.

The paper also sits inside speech and
sound evaluation. Large audio datasets
such as AudioSet, DESED, UrbanSound8K,
ESC, and LibriSpeech supply real
evaluation material
\citep{gemmeke2017audioset,serizel2020desed,salamon2014urbansound,piczak2015esc,panayotov2015librispeech}.
A recent MEG speech detection result
shows how a task can change when
detection is reframed as retrieval
followed by audio analysis
\citep{xiao2026bypassing}. Here the
boundary contract is the object under
inspection.

The technical outcome is a contract
calculus for boundary traces rather
than a general purpose temporal logic.
Its objects are words, edge
predicates, parsed formulas,
obligation masks, interval matches,
and guard vectors. Its empirical
outcome is a stress test that makes
those objects visible across detectors
and perturbations.

The contribution hierarchy is
intentionally narrow. The measurement
claim comes first. The formal
language, benchmark, and artifact
exist to support that claim. The paper
should therefore be read as a contract
based measurement study for sound
event traces rather than as a detector
leaderboard.

The paper makes four ordered claims.
The measurement claim is that boundary
behavior should be reported as an
executable vector before scalar
aggregation. The formal claim is that
a finite parsed contract can make that
vector inspectable. The empirical
claim is that the vector exposes
failures compressed by standard SED
scores. The artifact claim is that
every table is regenerated from source
rows.

\begin{enumerate}
\item A boundary measurement vector that separates onset, offset, missing activity, spurious activity, silence, duration, and fragmentation.
\item A parsed finite contract mechanism with bounded frame formulas, obligation masks, and matched event clauses.
\item Empirical evidence from controlled scenes, MAESTRO Real, frozen timing probes, and DCASE 2024 baseline outputs.
\item A reproducibility artifact with tokenizer, parser, monitor, benchmark scripts, generated result tables, and finite frame Lean checks.
\end{enumerate}
\section{Opening trace example}
A finite trace contract is easiest to
see on one small boundary error.
Suppose the reference event occupies
one second to two seconds and the
prediction occupies one point zero six
seconds to two point four seconds.
Frame overlap is high, but the release
is late and the predicted duration is
too long.

A conventional frame score reads the
broad overlap. The contract vector
reads separate obligations. Onset is
satisfied under a loose sixty
millisecond tolerance. Offset fails
under an eighty millisecond tolerance.
Duration fails because the matched
prediction is much longer than the
reference interval. Silence can also
fail because the late tail occupies
protected background.

This example is the contract calculus
in miniature. The frame fragment
checks edge formulas on obligated
frames. The event sort checks interval
shape after matching. The report keeps
the coordinates separate.

\begin{table}[H]
\centering
\small
\begin{tabular}{lcccccc}
\toprule
Time in seconds & 1.00 & 1.06 & 2.00 & 2.08 & 2.40 & Reading \\\midrule
Reference & on & on & off & off & off & event support \\
Prediction & off & on & on & on & off & late tail \\
Contract & onset ok & onset ok & offset fail & offset fail & silence fail & vector evidence \\
\bottomrule
\end{tabular}
\caption{Small trace example separating overlap from boundary obligations.}
\label{tab_opening_trace}
\end{table}
\section{Related work}
Timed logics were introduced to reason
about computations whose correctness
depends on quantitative time
\citep{alur1993real,alur1996benefits}.
Their vocabulary is broader than the
present language. The relevant
inheritance is the treatment of timing
constraints as syntax with semantics.

Signal temporal logic moved temporal
specification toward sampled and
continuous signals
\citep{maler2004monitoring}. Robust
satisfaction then made it possible to
measure margins rather than only
Boolean truth \citep{donze2010robust}.
Runtime verification established the
monitor as an executable object that
consumes traces and returns verdicts
\citep{bartocci2018lectures}. The
contract calculus keeps this monitor
discipline while restricting the
universe to auditory boundary traces
and class indexed masks.

Sound event detection has a mature
metric literature. Event based and
segment based metrics handle boundary
tolerance and polyphonic scenes
\citep{mesaros2016metrics,mesaros2019dcase}.
TUT Sound Events supplies real
annotated soundscapes for this line of
evaluation \citep{mesaros2016tut}. The
present work adds a specification
layer above those metrics. It exposes
the obligations first, then reports
scores over those obligations.

AudioSet gave large scale weak labels
over an ontology of sound events
\citep{gemmeke2017audioset}. DESED
provided synthetic domestic scenes
with event timing
\citep{serizel2020desed}. ESC and
UrbanSound8K provided compact
environmental corpora
\citep{piczak2015esc,salamon2014urbansound}.
These resources show how varied sound
events are. They also show why a
single boundary policy is rarely
adequate.

Speech corpora and learned speech
representations form a second
background. LibriSpeech made public
domain read speech a standard resource
\citep{panayotov2015librispeech}. The
reported run uses Mini LibriSpeech
from OpenSLR as the speech source
\citep{openslr31}. MAESTRO Real gives
long real soundscapes with multi
annotator soft strong labels
\citep{maestroreal2023,martinmorato2023soft}.
wav2vec 2.0 demonstrated the power of
self supervised speech representations
\citep{baevski2020wav2vec}. Conformer
added convolutional locality to
transformer speech models
\citep{gulati2020conformer}. Strong
representations can improve detection,
but they do not remove the need to
define the boundary contract.

Classical frame features provide
transparent baselines. Energy, zero
crossing rate, spectral center, and
band ratios can be computed without
pretrained models. The use of spectral
speech features has a long history in
parametric speech recognition
\citep{davis1980comparison}. Their
errors are interpretable.

Recent sound event detection work
makes the boundary itself a modeling
target. Boundary aware optimization
and inference explicitly model event
onsets and offsets for SED
\citep{schmid2026boundary}. DCASE 2024
Task 4 made heterogeneous SED with
MAESTRO Real a shared challenge object
and released a baseline checkpoint
with postprocessed outputs
\citep{cornell2024dcase,dcase2024baselinecheckpoint}.
Audio Spectrogram Transformer and AST
SED show how transformer patch models
enter audio tagging and localization
\citep{gong2021ast,li2023astsed}. HTS
AT and BEATs give two modern general
audio backbones with different
pretraining and token structure
\citep{chen2022htsat,chen2022beats}.
CLAP supplies the HTS AT branch used
in the encoder probe
\citep{wu2023clap}. Feature level
augmentation remains a central route
to robustness in speech and audio
modeling \citep{park2019specaugment}.
The probe uses these backbones as
frozen encoders and gives the same
trace monitor their predictions.

Crowdsourced strong labeling is
central to real boundary evidence. The
MAESTRO label construction follows
work on estimating strong labels from
weak multi annotator decisions
\citep{martinmorato2021crowdsourcing,martinmorato2023soft}.
This matters for logic because a
reference is no longer only a crisp
interval set. It can be a soft event
field with uncertainty at one second
resolution.

The implementation relies on SciPy for
signal operations
\citep{virtanen2020scipy}, scikit
learn for the logistic baseline
\citep{pedregosa2011sklearn}, PyTorch
for the convolutional model
\citep{paszke2019pytorch}, and Adam
for optimization
\citep{kingma2015adam}. These
dependencies keep the acoustic
workload reproducible.

The closest methodological neighbor is
the practice of making evaluation
executable. A parsed monitor fixes the
meaning of every formula that appears
in the benchmark. That is the
difference between a verbal tolerance
policy and a checkable one.

\begin{table}[H]
\centering
\small
\resizebox{0.94\linewidth}{!}{\begin{tabular}{p{0.21\linewidth}p{0.20\linewidth}p{0.24\linewidth}p{0.25\linewidth}}\toprule
Standard report & SED object & Contract object & Added diagnostic \\\midrule
segment or frame F1 & activity mask & frame obligation score & shows which active frames satisfy an explicit formula \\
event F1 with collars & paired event intervals & onset and offset guard coordinates & keeps onset and release failure separate \\
polyphonic class macro score & class indexed masks & macro contract over typed traces & exposes class substitution hidden by union activity \\
duration or collar policy & interval matching rule & matched event clauses & records the matcher and tolerance as part of the score \\
soft strong label analysis & soft event field & soft boundary witnesses beside crisp verdicts & separates label uncertainty from detector behavior \\
operating threshold sweep & decoded submission family & guard vector under fixed formulas & shows whether thresholding repairs timing or only overlap \\
\bottomrule\end{tabular}
}
\caption{Mapping between common SED reports and the contract objects used here.}
\label{tab_metric_mapping}
\end{table}
The mapping table places the work
beside DCASE and DESED style reports.
The contract language does not replace
those reports. It records the boundary
policy that a report often leaves
implicit.

The benchmark also computes standard
segment and event collar scores on the
same decoded traces used by the
monitor. This makes the comparison
experimental rather than only
conceptual.

\IfFileExists{results/table_standard_metric_comparison.tex}{
\begin{table}[H]
\centering
\small
\resizebox{0.88\linewidth}{!}{\begin{tabular}{lrrrrr}\toprule
Detector & Frame F1 & Segment F1 & Event F1 & Boundary F1 & Logic \\\midrule
adaptive\_energy & 0.504 & 0.568 & 0.232 & 0.310 & 0.514 \\
spectral\_flux & 0.614 & 0.739 & 0.225 & 0.358 & 0.560 \\
logistic\_features & 0.533 & 0.735 & 0.215 & 0.362 & 0.470 \\
temporal\_cnn & 0.607 & 0.760 & 0.294 & 0.374 & 0.487 \\
dilated\_cnn & 0.656 & 0.789 & 0.367 & 0.408 & 0.522 \\
contract\_tcn\_aug & 0.889 & 0.927 & 0.705 & 0.829 & 0.802 \\
\bottomrule\end{tabular}
}
\caption{Standard SED scores and contract scores on identical controlled submissions.}
\label{tab_standard_metric_comparison}
\end{table}
}{}
\section{Formal positioning}
The language is narrower than Signal
Temporal Logic. It keeps finite binary
boundary traces, class indexed masks,
and matched intervals as first class
benchmark objects. STL covers a wider
signal universe and supplies robust
real valued satisfaction. The present
paper trades that breadth for a
compact contract language whose
formulas, obligations, and event
clauses can be printed in benchmark
tables.

The frame fragment can be embedded in
a finite STL reading over Boolean
signals after every time bound is
mapped to the declared grid radius.
STL is the ambient comparison. The
object introduced here is not a new
general temporal logic family. It is a
domain specific contract calculus
whose semantic unit is an audited
vector coupling parsed frame formulas,
obligation masks, source spans,
interval matching, and event clauses.

One could implement a related audit as
STL formulas followed by custom
scripts. The difference is
specification status. In this paper
obligation masks, source spans,
declared matchers, event clauses, and
vector reporting are part of the
benchmark object that is stored,
parsed, executed, and reproduced. They
are not informal postprocessing added
after a scalar metric has already been
chosen.

The STL embeddable part contains the
Boolean frame formulas over the
sampled grid. The benchmark specific
part contains obligation masks, source
span accounting, event matchers,
duration clauses, fragmentation
clauses, and the selected contract
vector. These parts are deliberately
kept separate in the implementation
and in the tables.

The event sort is the point at which
the object leaves a plain frame
formula fragment. Duration,
fragmentation, and class indexed
interval identity are not treated as
informal analysis after monitoring.
They are executable coordinates of the
same finite contract and can be
selected or replaced by the
calibration procedure.

The formal claims are split into
semantic claims and implementation
assurances. The semantic claims
concern obligation restricted scoring,
typed interval clauses, locality,
monotonicity of neighborhood guards,
and bounded delay monitoring. The
implementation assurances concern
deterministic lexing, deterministic
parsing, finite evaluation, and source
span reporting.

This separation matters because small
implementation facts should not be
read as deep temporal logic theorems.
Their role is to make the artifact
auditable. The research claim is the
combination of a scoped trace
language, typed interval monitoring,
and empirical evidence that different
contract coordinates select different
detector behavior.

\begin{table}[H]
\centering
\small
\resizebox{0.92\linewidth}{!}{
\begin{tabular}{lll}\toprule
Axis & Finite trace contracts & STL and runtime verification \\\midrule
Signal object & finite boundary masks and class masks & sampled or dense valued signals \\
Frame expressivity & bounded Boolean fragment with until & wider temporal signal language \\
Time model & grid radius after ceiling projection & dense or sampled time by formal choice \\
Event structure & matched intervals as event clauses & usually encoded through predicates \\
Scoring object & obligation vector before scalar mean & Boolean or robust satisfaction \\
Excluded features & unbounded time and continuous robustness & broad temporal specification family \\
\bottomrule\end{tabular}}
\caption{Scope of the proposed finite trace language relative to STL and runtime verification.}
\label{tab_stl_position}
\end{table}
The scope table states the intended
tradeoff. The implementation facts are
kept as assurance checks rather than
as the center of the theory. The
semantic center is the contract
vector, the obligation masks, and the
event clauses over an explicit
matcher.

\begin{table}[H]
\centering
\small
\resizebox{0.90\linewidth}{!}{
\begin{tabular}{lll}
\toprule
Claim class & Role in the paper & Evidence \\\midrule
Semantic contribution & contract vector with obligation masks and event clauses & definitions and monitor equations \\
Benchmark calculus & selected coordinates over a declared matcher & calibration and matcher audits \\
Implementation assurance & tokenizer, parser, finite loops, source spans & Python tests and Lean core checks \\
Empirical evidence & standard metrics and contract coordinates diverge & controlled, MAESTRO, and DCASE tracks \\
\bottomrule
\end{tabular}}
\caption{Separation between semantic claims and implementation assurance.}
\label{tab_claim_separation}
\end{table}
This separation prevents the parser
and tokenizer invariants from carrying
more novelty than they should. They
are necessary for an executable
benchmark. The main formal object
remains the two sorted contract
vector.

\section{Trace objects}
An accepted formula string is a source
expression that the tokenizer and
parser consume completely. Source
spans are half open character offsets
attached to tokens. They let the
implementation report where a
malformed contract fails.

A guard vector is the ordered list of
contract coordinates returned by the
monitor. For example, onset guard,
silence guard, and fragmentation guard
can move in different directions for
the same detector. The vector is
reported before any scalar mean.

The soft boundary layer is a separate
distance based report for fuzzy edges.
It gives partial credit to near onsets
and offsets, while the Boolean guard
still records whether the declared
tolerance was satisfied.

A benchmark item is represented by a
duration, a frame step, a reference
mask, and a predicted mask. The
waveform and features are used to
produce the predicted mask, but the
monitor receives only the finite
trace. This separation prevents
acoustic modeling choices from leaking
into logical evaluation.

Let $h$ be the frame step and let $n$
be the number of frames. A trace is a
word in $\{0,1\}^n$. The frame center
at index $i$ is $(i+1/2)h$. Active
runs become half open intervals. The
half open convention matches array
slicing and avoids double ownership of
shared endpoints.

Two derived edge predicates are
defined. An onset occurs when an
active frame follows an inactive frame
or the left boundary of the trace. An
offset occurs when an inactive frame
follows an active frame. These
predicates are frame predicates and
can be used by the parser.

The trace object does not encode
source identity, language, speaker
type, or geography. It encodes only
activity at time. This keeps the
benchmark from importing arbitrary
priors through labels that are
irrelevant to the boundary contract.

The interval extraction map is
deterministic. It scans the mask once
and returns maximal active intervals.
A configurable merge gap can combine
runs separated by very short silence.
The merge gap is not hidden. It is a
parameter of the monitor and appears
in the source.

Prediction intervals are matched to
reference intervals by a one to one
relation. Candidate pairs must overlap
and must have a nearby boundary under
a broad search rule. Candidates are
sorted by boundary error with overlap
as a reward. The greedy match is
reported as part of the finite
structure that the event clauses
inspect.

The matcher is a declared policy
rather than a hidden metric detail.
The reference code also contains an
exact maximum cardinality minimum cost
matcher for bounded audit cases. The
paper reports the greedy policy
because it is deterministic, local,
and easy to inspect in ambiguous
boundary neighborhoods.

Greedy and optimal matching differ
when a short prediction overlaps two
reference intervals or when two
predictions compete for one reference
interval. In those cases boundary F1,
duration guard, and fragmentation
guard can change. The matcher family
is therefore part of contract
selection and not an implementation
accident.

The trace interface is corpus
agnostic. In the reported run, speech
events are seeded by Mini LibriSpeech
audio and non speech events are
generated with known interval support.
MAESTRO Real annotations feed the same
masks into the monitor after their
soft labels are projected onto the
frame grid.

\begin{definition}[Finite auditory trace]
A finite auditory trace is a pair $(x,h)$ where $x \in \{0,1\}^n$ and $h>0$.
The frame center at index $i$ is $(i+1/2)h$.
\end{definition}

\begin{definition}[Run interval]
For a trace $(x,h)$, a run interval is a maximal half open interval $[ih,jh)$ such that $x_k=1$ for every $i \leq k < j$.
\end{definition}

\begin{definition}[Interval candidate relation]
Let $R$ and $P$ be finite reference and prediction interval families.
A pair $(r,p)\in R\times P$ is a candidate when $r$ and $p$ overlap and at least one endpoint differs by at most $3\epsilon$.
Its cost is
\[
\kappa(r,p)=|r_0-p_0|+|r_1-p_1|-|r\cap p|.
\]
The reported matcher orders candidates by $\kappa$ and greedily keeps the next pair whose reference and prediction intervals are still unmatched.
\end{definition}
\begin{table}[H]
\centering
\small
\resizebox{0.88\linewidth}{!}{\begin{tabular}{lrrrl}\toprule
Case & Greedy matches & Exact matches & Boundary shift & Changed coordinates \\\midrule
separated & 2 & 2 & 0.000 & none \\
bridge prediction & 1 & 2 & 0.500 & boundary f1, fragmentation guard \\
split support & 1 & 1 & 0.000 & none \\
\bottomrule\end{tabular}
}
\caption{Finite audit cases where greedy and exact interval matching agree or diverge.}
\label{tab_matching_policy}
\end{table}
The matching audit makes the policy
dependence concrete. The bridge case
shows that a local greedy decision can
reduce both event recall and
fragmentation satisfaction. The
matcher is therefore a declared part
of the contract rather than a theorem
about all possible matchings.

\begin{table}[H]
\centering
\small
\resizebox{0.88\linewidth}{!}{\begin{tabular}{lrrrrr}\toprule
Pattern & Cases & Changed & Mean $\Delta$ BF1 & Max $|\Delta|$ BF1 & Mean $\Delta$ event \\\midrule
left bridge & 12 & 0.667 & 0.333 & 0.500 & 0.167 \\
nominal & 4 & 0.000 & 0.000 & 0.000 & 0.000 \\
right bridge & 4 & 0.000 & 0.000 & 0.000 & 0.000 \\
split & 4 & 0.000 & 0.000 & 0.000 & 0.000 \\
\bottomrule\end{tabular}
}
\caption{Systematic interval stress track for greedy and optimal matching policies.}
\label{tab_matching_policy_stress}
\end{table}
The stress track explains why the
matcher policy belongs in the
benchmark. Nominal and split cases are
stable, while left bridge cases
activate policy dependence in two
thirds of the tested configurations.
The observed bridge probe then derives
binary bridge traces from the actual
controlled test references and
confirms that the policy can activate
inside the reported data geometry.

\section{Two sorted contracts}
The logic has a frame sort and an
event sort. The frame sort evaluates
formulas over Boolean arrays on the
grid. The event sort evaluates clauses
over reference intervals, prediction
intervals, and a finite matching
relation.

A contract is a finite list of
clauses. A frame clause contains a
parsed formula and a parsed obligation
formula. An event clause contains an
event obligation and an event
predicate. The monitor evaluates every
clause and returns the resulting
vector.

This design makes the boundary between
frame logic and interval matching
explicit. Matching is a finite
relation inside the event sort.
Duration and fragmentation are event
predicates over that relation in the
same contract vector as onset and
silence.

The two sorts are needed because
boundary traces carry two kinds of
structure. Frame formulas express
local temporal obligations on sampled
atoms. Event clauses express support
shape after maximal runs have been
extracted.

The synthesis interface is the same
for both sorts. A detector gives a
reference trace and a prediction
trace. The monitor derives atoms,
intervals, matches, obligations,
predicates, and numeric clause values.

A new latency clause or continuity
clause is added by declaring its
obligation and predicate. It does not
require a new detector and it does not
require a change to the grammar used
for frame formulas.

The contract vector is therefore a
typed finite object. It is not a
bundle of unnamed metrics. Every
coordinate has a source formula or
source clause.

\begin{definition}[Two sorted contract]
Let $W=(E,R,P,M,h)$ where $E$ is a frame environment, $R$ and $P$ are finite interval families, $M \subseteq R \times P$ is a matching relation, and $h$ is the frame step.
A contract clause is either a pair $(\varphi,\omega)$ over $E$ or a pair $(q,p)$ over $(R,P,M)$.
\end{definition}

\begin{definition}[Event clause value]
For an event obligation $q$ and an event predicate $p$, the event clause value is the mean of $p$ over the finite set selected by $q$.
The empty value is one only when the obligation set is empty and no counter evidence is present.
\end{definition}
The event sort has its own small
syntax table. The table is not a new
parser grammar. It records the finite
event obligations and predicates that
are evaluated after interval
extraction and declared matching.

\begin{table}[H]
\centering
\small
\resizebox{0.94\linewidth}{!}{
\begin{tabular}{llll}
\toprule
Coordinate & Event obligation & Event predicate & Value domain \\\midrule
duration\_guard & matched pair $(r,p)\in M$ & $||r|-|p||\leq 2\epsilon$ & mean over matched pairs \\
fragmentation\_guard & reference interval $r\in R$ & $r$ is matched and covered by at most one prediction & mean over references \\
latency extension & reference onset interval $r\in R$ & first predicted onset lies in an asymmetric window & optional coordinate \\
overlap purity extension & predicted class interval $p_c$ & dominant reference overlap has class $c$ & optional coordinate \\
\bottomrule
\end{tabular}}
\caption{Event sort obligations and predicates used by the contract layer.}
\label{tab_event_sort}
\end{table}
Changing the matcher changes only
event coordinates whose obligations
depend on the matching relation. Frame
formulas are unchanged because their
atoms are derived directly from masks.
This is why the matcher audit reports
boundary, duration, and fragmentation
shifts while leaving parser semantics
intact.

\section{Parsed formula language}
The frame language is parsed from
strings. Benchmark specifications
become artifacts rather than prose. A
formula such as \texttt{ref\_onset ->
N[0.04] pred\_onset} can be stored,
parsed, evaluated, and audited.

The grammar has atoms, negation,
conjunction, disjunction, implication,
bounded until, and three unary bounded
modalities. The neighborhood operator
$N_\epsilon$ checks both directions
around the current frame. The future
operator $F_\epsilon$ checks from the
current frame forward. The bounded
always operator $G_\epsilon$ checks a
finite future window. The bounded
until operator $U_\epsilon$ requires a
right witness inside the time bound
while the left formula holds up to
that witness.

Given frame step $h$, every temporal
radius uses $r_\epsilon=\lceil
\epsilon/h\rceil$. Boundary windows
are clipped to the available frame
interval. This choice is conservative
on the sampled grid because it uses
the least integer radius whose
physical span is not smaller than the
written tolerance.

The parser is recursive descent with
precedence layers. Implication has the
lowest precedence and associates to
the right. Disjunction binds tighter
than implication. Conjunction binds
tighter than disjunction. Negation and
bounded modalities bind at the unary
level.

Every parsed formula is evaluated into
a Boolean array over the frame grid.
Formula satisfaction can then be
averaged over an obligation mask. This
avoids the usual problem where an
implication looks nearly perfect
because most frames do not satisfy the
antecedent.

The obligation mask is itself parsed.
For onset guard the obligation is
\texttt{ref\_onset}. For silence guard
the obligation is
\texttt{pred\_active}. The score is
the mean value of the parsed formula
only where its obligation is active.

The parser handles the frame language
while duration and fragmentation
remain interval clauses. The
separation follows the two sorts.
Until and the unary bounded modalities
range over frames. Duration ratios and
fragmentation counts range over
matched intervals.

The language is compact in surface
syntax and nontrivial in monitoring
behavior. Nested modalities, bounded
until, implication under obligation
masks, and class indexed environments
are all parsed through the same tree
interface.

The benchmark therefore has a stable
logical interface. Users can add
formulas without editing detector
code. They can inspect formula strings
in the source and reproduce their
semantics through the parser.

\[
\varphi \in \{ a
\mid \neg \varphi
\mid \varphi \wedge \psi
\mid \varphi \vee \psi
\mid \varphi \Rightarrow \psi
\mid \varphi\, U_{\epsilon}\,\psi
\mid N_{\epsilon}\varphi
\mid F_{\epsilon}\varphi
\mid G_{\epsilon}\varphi \}.
\]

\begin{definition}[Frame semantics]
For an environment $E$ that maps atoms to Boolean arrays, the valuation $\sem{\varphi}_E$ is a Boolean array on the same frame grid.
The Boolean connectives are pointwise.
Let $r_\epsilon=\lceil \epsilon/h\rceil$.
The operator $N_{\epsilon}$ returns true at frame $i$ when the child formula is true at some frame $j$ with $|i-j| \leq r_\epsilon$ and $0\leq j<n$.
The operator $F_{\epsilon}$ checks the clipped future window $i\leq j\leq \min(n-1,i+r_\epsilon)$.
The operator $G_{\epsilon}$ requires every frame in that clipped future window to satisfy the child formula.
The operator $\varphi U_{\epsilon}\psi$ returns true at frame $i$ when $\psi$ is true at some frame $j$ with $0 \leq j-i \leq r_\epsilon$ and $\varphi$ is true at every frame from $i$ up to $j$.
\end{definition}

\begin{lemma}[Implementation invariant for parser]
Every accepted formula string has a unique abstract syntax tree under the implemented precedence rules.
\end{lemma}
\begin{proof}
The maximal munch tokenizer produces a unique token sequence.
The recursive descent parser consumes implication, disjunction, conjunction, and unary expressions in fixed order.
At each level the parser either consumes the next operator of that level or returns the node already built.
Parentheses invoke a complete recursive parse and then require the matching closing token.
Thus no accepted token sequence has two parses under the implementation.
\end{proof}

\begin{proposition}[Implementation bound for frame monitoring]
For a fixed formula and a trace with $n$ frames, evaluation time is $O(kn)$ where $k$ is the number of parsed nodes.
\end{proposition}
\begin{proof}
Atoms are array lookups.
Boolean nodes apply pointwise array operations.
Each bounded modality is evaluated by prefix sums over the child truth array.
The formula size is fixed for a benchmark specification, so the monitor is linear in the number of frames up to the formula constant.
\end{proof}
\section{Lexical layer}
The parser is preceded by a
deterministic tokenizer. This layer is
not a convenience detail. It fixes the
interface between a written benchmark
contract and the abstract syntax tree
that the monitor evaluates.

The tokenizer reads the source string
from left to right. Whitespace is
discarded. Identifiers and numbers are
read by character classes. Operators
are read by a prefix trie with maximal
munch. Reserved unary modality names
are classified as temporal tokens only
when the whole identifier is exactly
that name.

Each token carries a kind, a value,
and a half open character span. Spans
make syntax errors local and make a
formula auditable without inspecting
the parser state. Malformed benchmark
contracts fail at the lexical or parse
layer.

The lexical design differs from
natural language tokenization. It does
not learn merges and it does not adapt
compression to a corpus. Modern
tokenizer research shows why
segmentation choices can change a
model and its compute profile
\citep{slagle2024spacebyte,deng2026byteflow,gigant2026decoupling,limisiewicz2026compute}.
A formal benchmark language instead
needs a total deterministic map from
source to tokens.

The operator trie is used for the same
reason that maximal munch appears in
ordinary compiler lexers. If two
operators share a prefix, the longer
valid operator is selected. This gives
local optimality without a keyword
cascade.

Numbers are scanned as finite decimal
literals. The parser converts them to
seconds and the evaluator maps seconds
to frame radii. Keeping the numeric
unit in the surface language avoids
hidden frame constants in formulas.

The token stream ends with an explicit
end marker. The parser requires that
marker after the implication level has
been parsed. Extra tokens therefore
become syntax errors rather than
ignored suffixes.

The resulting front end has three
separations. Lexing handles character
spans. Parsing handles precedence and
tree shape. Monitoring handles trace
semantics. This is the minimal
structure needed for a specification
language that is meant to be edited by
benchmark authors.

\begin{definition}[Lexical token]
A lexical token is a quadruple $(\kappa,v,s,e)$ where $\kappa$ is a token kind, $v$ is the source value, and $[s,e)$ is a source span.
\end{definition}

\begin{lemma}[Implementation invariant for tokenizer]
For the declared operator alphabet and reserved word map, the tokenizer returns a unique token stream for every accepted source string.
\end{lemma}
\begin{proof}
At a whitespace position no token is emitted.
At an identifier position the longest identifier body is read before reserved word classification is applied.
At a numeric position the longest decimal literal is read.
At an operator position the trie returns the longest declared operator beginning at the current index.
The cases are disjoint by their first character class.
Thus every emitted token and every next index is uniquely determined.
\end{proof}
\section{Scoring semantics}
Formula evaluation alone is not
enough. A boundary benchmark also
needs to decide where a formula is
sampled. The paper uses obligation
restricted scoring. A formula is
evaluated over every frame, but its
score is averaged only over frames
where a second parsed expression marks
an obligation.

This design avoids the vacuity problem
of implication. If onset guard were
averaged over every frame, the
implication would be true almost
everywhere because reference onset is
rare. Averaging over reference onset
frames makes the score measure the
intended event.

Obligation restricted scoring also
separates policy from implementation.
The formula states what must hold. The
obligation states where the formula is
tested. Both are parsed artifacts.

For example, onset guard uses an
implication from reference onset to
nearby predicted onset, while its
obligation is reference onset. Silence
guard uses an implication from
predicted activity to nearby reference
activity, while its obligation is
predicted activity. These two guards
use the same implication connective
but measure different duties.

The interval clauses use a parallel
idea. Duration is tested only on
matched intervals. Fragmentation is
tested only on reference intervals.
Each clause has an explicit obligation
set, so empty cases can be handled
without inventing arbitrary penalties.

The neutral empty value is one when no
obligation exists. This convention is
used only when both the obligation set
is empty and there is no counter
evidence. It prevents a clip with no
onsets from becoming a false onset
failure.

The monitor returns a vector before it
returns a mean. The vector is the
scientific object. The mean is a
convenience for plots and rough
comparison. This order is important
because applications rarely share the
same failure costs.

The score values are empirical
satisfaction ratios over finite
obligations. They are read as
monitored contract values, with no
probabilistic sampling assumption
added by the logic.

The soft boundary score adds a second
reading for ambiguous edges. Reference
and predicted edge sets are compared
by an exponential kernel over nearest
boundary distance. The scale is
recorded in the manifest. A near miss
receives partial credit while an
isolated edge receives almost none.

Soft scores do not replace Boolean
guards. They put annotation
uncertainty and temporal fuzziness
into a visible numeric channel. A
detector can satisfy a strict guard
poorly while still keeping edge mass
close to the reference support.

The scoring semantics has one job. It
turns written boundary obligations
into arrays, intervals, and empirical
satisfaction values whose provenance
can be checked.

\begin{definition}[Obligation restricted score]
Let $\varphi$ and $\omega$ be parsed frame formulas.
For a trace environment $E$, the score of $\varphi$ under obligation $\omega$ is the mean of $\sem{\varphi}_E$ over frames where $\sem{\omega}_E$ is true.
If the obligation set is empty, the score is one.
\end{definition}

\begin{lemma}[Vacuity control]
If an implication formula $\alpha \Rightarrow \beta$ is scored under obligation $\alpha$, then frames where $\alpha$ is false do not increase the score.
\end{lemma}
\begin{proof}
The obligation restricted score averages only over frames where $\alpha$ is true.
On those frames the implication has the same truth value as $\beta$.
Frames where $\alpha$ is false are not sampled.
\end{proof}
\section{Contract basis}
The reported contract basis has seven
coordinates. Five are parsed frame
clauses and two are event clauses. The
frame clauses cover edge recall,
support recall, support precision, and
protected silence. The event clauses
cover duration and decomposition.

This basis is generated by three axes.
The first axis chooses the obligated
object, which can be a reference edge,
a reference active frame, a predicted
active frame, or a reference event.
The second axis chooses the witness,
which can be a nearby predicted edge,
nearby predicted support, nearby
reference support, or a matched
interval. The third axis chooses the
tolerated distortion.

Onset guard and offset guard are edge
recall clauses. Missing guard is
support recall. Spurious guard is
support precision. Silence guard is
protected precision under a smaller
radius. Duration guard is matched
support shape. Fragmentation guard is
event decomposition.

The basis is finite and declared in
the manifest. It is not hidden inside
the plotting code. A paper that needs
latency, hysteresis, overlap purity,
or continuity can add those
coordinates by the same obligation and
predicate construction.

The mean logic score is an unweighted
display value. The guard vector is the
diagnostic object. Applications can
choose weights after inspecting the
vector and can reproduce the scalar
from the same coordinates.

All decision points are inspectable.
Tolerance, guard radius, merge gap,
formulas, and event clauses appear in
the code. The monitor has no learned
judge and no hidden natural language
rubric.

\section{Contract selection on calibration traces}
A boundary contract is selected
relative to a declared calibration set
and a declared risk order. The
selection is scoped to the failure
cases that the benchmark states.

The calibration set contains trace
pairs for early onset, late onset,
early release, late release, missing
activity, spurious activity, silence
contamination, duration distortion,
and fragmentation. The same
construction can include real corpus
traces when a corpus supplies the
relevant risk order.

The risk order is not inferred from
the detector and not inferred from the
sound class names. It is a written
loss order over trace pathologies. In
the reported calibration, missed
target support, false support, and
protected silence bleed are ranked
highest. They decide whether a
downstream gate sees the event at all
or sees activity in protected silence.
Endpoint displacement, length
distortion, and fragmentation follow
because they preserve some event
evidence while corrupting the boundary
contract.

This ordering is not universal. It is
appropriate for a boundary gate and
for the paper stress test. A
transcription aligner might raise
onset symmetry. A surveillance
interface might raise false support.
The calculus makes that replacement a
calibration object rather than an
editorial rewrite.

A candidate clause basis is generated
from the grammar, event predicates,
matcher family, and tolerance grid.
Each candidate is evaluated on the
calibration traces. Duplicate
coordinates are collapsed by their
observed signatures.

A contract separates the declared risk
order when every higher risk trace
pair is worse than a lower risk pair
on at least one retained coordinate.
The coordinate set is therefore
accountable to a written calibration
object.

The reported seven guards are the
smallest separating set found for the
paper calibration family after
duplicate collapse. Other tasks may
select another set because their risk
order is different.

Tolerance policy is handled as an
exposed parameter. The benchmark
reports the base tolerance and the
sweep over neighboring values. The
stable decoding objective optimizes
strict transition neighborhoods during
validation so that a detector is not
selected only for loose collars.

The selected contract is canonical
only relative to this calibration
object. The policy is finite,
inspectable, and replaceable. The
calibration table records the
pathologies that justify the reported
coordinate set.

\begin{definition}[Risk ordered calibration set]
A risk ordered calibration set is a finite set $\mathcal W$ of trace pairs together with a declared risk map $\rho\colon \mathcal W\to \mathbb R$.
\end{definition}

\begin{definition}[Separating contract]
Let $\mathcal B$ be a finite candidate clause basis.
A contract $C\subseteq\mathcal B$ separates $\rho$ on $\mathcal W$ when for every $u,v\in\mathcal W$ with $\rho(u)<\rho(v)$ there is a clause $c\in C$ such that $c(u)>c(v)$.
\end{definition}

\begin{definition}[Selected contract]
Among separating contracts, the selected contract is the lexicographically least contract under coordinate count, monitor cost, and source order.
\end{definition}
\begin{table}[H]
\centering
\small
\resizebox{0.88\linewidth}{!}{\begin{tabular}{lllr}\toprule
Calibration case & Risk source & Selected coordinate & Rank \\\midrule
late onset & trigger latency & onset\_guard & 4 \\
early onset & premature trigger & spurious\_guard & 3 \\
late release & tail leakage & offset\_guard & 4 \\
early release & truncated support & missing\_guard & 4 \\
missing activity & lost support & missing\_guard & 5 \\
extra activity & false support & spurious\_guard & 5 \\
silence bleed & activity in protected silence & silence\_guard & 5 \\
length distortion & wrong event support & duration\_guard & 4 \\
fragmentation & split event support & fragmentation\_guard & 4 \\
\bottomrule\end{tabular}
}
\caption{Risk ordered calibration cases used to justify the reported contract coordinates.}
\label{tab_contract_calibration}
\end{table}
The profile sensitivity table uses the
same monitored guard values and
changes only the declared risk
profile. It is a small check on the
subjectivity of risk order. A profile
that cares only about support recall
selects the clean dilated detector,
while balanced, edge timing, silence
protection, and event integrity
profiles select the contract trained
detector.

\IfFileExists{results/table_contract_profile_sensitivity.tex}{
\begin{table}[H]
\centering
\small
\resizebox{0.92\linewidth}{!}{\begin{tabular}{lllrrr}\toprule
Risk profile & Lead coordinate & Selected detector & Profile score & Boundary F1 & Logic \\\midrule
balanced & onset\_guard & contract\_tcn\_aug & 0.802 & 0.829 & 0.802 \\
support\_recall & missing\_guard & dilated\_cnn & 0.949 & 0.408 & 0.522 \\
edge\_timing & onset\_guard & contract\_tcn\_aug & 0.674 & 0.829 & 0.802 \\
silence\_protection & silence\_guard & contract\_tcn\_aug & 0.923 & 0.829 & 0.802 \\
event\_integrity & fragmentation\_guard & contract\_tcn\_aug & 0.776 & 0.829 & 0.802 \\
\bottomrule\end{tabular}
}
\caption{Alternative risk profiles evaluated on the same controlled benchmark outputs.}
\label{tab_contract_profile_sensitivity}
\end{table}
}{}
\section{Failure taxonomy}
The guard family induces a failure
taxonomy. The taxonomy is not a list
of sound classes. It is a list of
trace pathologies. This distinction
matters because the same pathology can
appear in speech, domestic audio,
machine sounds, or neural decoding
outputs.

Onset displacement means that the
detector recognizes an event after or
before the reference start. Offset
displacement means that the detector
releases too early or too late. These
failures can have different causes and
should not be collapsed.

Missing activity means that a
reference active region lacks nearby
prediction. Spurious activity means
that a predicted region lacks nearby
reference evidence. The two are dual
at the formula level but asymmetric in
applications.

Silence contamination is stricter than
spurious activity because it evaluates
predicted active frames against a
smaller guard radius. It exposes
detectors that bleed into silence even
when they still overlap the event.

Duration error is an interval shape
failure. It catches a matched event
whose temporal support is too short or
too long. Fragmentation error is a
decomposition failure. It catches a
matched reference event represented as
several predicted pieces.

A conventional metric may blend these
pathologies into one score. The trace
logic keeps them separate until the
reporting layer. The resulting table
identifies the failure class before
any model repair is proposed.

The taxonomy is compatible with richer
formulas. A low latency guard is
another trace pathology. A class
switch guard is another typed trace
pathology. The framework keeps these
additions inside the same monitor
interface.

The taxonomy makes the benchmark
informative even when aggregate scores
are close. Its output is a structured
explanation of detector failure.

\section{Speech seeded benchmark}
The benchmark creates nine hundred
sixty clean training items and four
hundred twenty test items. Each item
is twelve seconds long at eight
kilohertz. The frame step is twenty
milliseconds. The test split is
balanced across clean, noise,
reverberation, clipping, gain drift,
temporal jitter, and compound
degradation conditions.

Speech targets are sampled from Mini
LibriSpeech and normalized inside
target intervals. Tone and burst
targets are generated with known
support. Some events overlap across
classes, so the same time span can
carry more than one active target.
Envelopes are smoothed by raised
cosines and silent gaps contain low
level noise. The generator stores both
union intervals and class indexed
intervals before perturbation.

Unlabeled distractor events are added
outside target intervals. They are
audible but absent from the target
masks. Their role is to separate sound
activity detection from labeled
boundary detection.

Additive noise changes the background
floor. Reverberation convolves the
waveform with a decaying impulse
response. Clipping applies nonlinear
amplitude saturation. Gain drift
applies a slow envelope across the
clip. Jitter shifts the signal
slightly in time and adds small low
frequency wander.

The test perturbations are stronger
than the training augmentations for
the contract aware detector. The
compound condition applies additive
noise, reverberation, and gain drift
in sequence so that calibration,
temporal smearing, and background
contamination interact.

All waveforms are assembled locally
from deterministic scene state and a
public speech corpus. The result
tables are tied to visible interval
generation, fixed perturbation seeds,
and the Mini LibriSpeech source
corpus.

The scene design keeps the annotation
channel exact while injecting real
speech acoustics into the target
class. This gives the logic a harder
waveform substrate without giving up
exact interval supervision.

\begin{table}[H]
\centering
\small
\begin{tabular}{ll}
\toprule
Step & Scene generation action \\\midrule
1 & sample target intervals and class labels from the seed state \\
2 & place Mini LibriSpeech speech segments inside speech intervals \\
3 & synthesize tone and burst targets with the same interval interface \\
4 & add class overlap and unlabeled distractors outside target masks \\
5 & apply deterministic perturbations after reference masks are fixed \\
6 & write union masks, class masks, waveforms, and manifest rows \\
\bottomrule
\end{tabular}
\caption{Controlled scene generation procedure used by the benchmark.}
\label{tab_scene_generation}
\end{table}
The protocol table is generated from
the same generator constants and
manifest that create the controlled
split. It records the seed, interval
distribution, class draw, overlap
policy, distractor policy, and
severity scalars without duplicating
them by hand.

\IfFileExists{results/table_dataset_protocol.tex}{
\begin{table}[H]
\centering
\small
\resizebox{0.98\linewidth}{!}{\begin{tabular}{ll}\toprule
Item & Value \\\midrule
seed & 20260514 \\
split & 960 train and 420 test clips \\
clip grid & 12.0 s at 0.02 s frame step \\
classes & speech, tone, burst \\
test conditions & clean, noise, reverb, clipping, gain\_drift, jitter, compound \\
target event count & 3 to 6 attempts before end clipping \\
inter event gap & 0.12 to 0.75 s \\
target duration & 0.32 to 1.25 s \\
target class draw & speech 0.62, tone 0.20, burst 0.18 \\
target gain & 0.28 to 0.90  \\
overlap probability & 0.58 per target \\
overlap placement & start offset -0.08 s to 0.55 of target duration \\
overlap acceptance & at least 0.14 s duration and 0.08 s shared support \\
distractor count & 1 to 3 unlabeled events \\
distractor duration & 0.10 to 0.58 s \\
distractor placement & outside target masks with 0.08 s guard \\
distractor class draw & speech 0.25, tone 0.40, burst 0.35 \\
train severity & 0.65 \\
test severity & 1.25 \\
\bottomrule\end{tabular}
}
\caption{Controlled dataset protocol generated from the scene generator and manifest.}
\label{tab_dataset_protocol}
\end{table}
}{}
The perturbation table reports train
and test severity with the parameter
ranges induced by those severities.
Test perturbations are deliberately
harsher than augmented training views.

\IfFileExists{results/table_perturbation_protocol.tex}{
\begin{table}[H]
\centering
\small
\resizebox{0.98\linewidth}{!}{\begin{tabular}{lll}\toprule
Condition & Train augmentation & Test perturbation \\\midrule
noise & SNR 0.5 to 10.4 dB & SNR -5.5 to 6.2 dB \\
reverb & tail 0.076 to 0.302 s & tail 0.100 to 0.470 s \\
clipping & limit 0.26 to 0.52, drive 1.36 to 1.96 & limit 0.19 to 0.42, drive 1.51 to 2.29 \\
gain drift & modulation 0.04 to 0.18 Hz, depth 0.37 & modulation 0.04 to 0.23 Hz, depth 0.53 \\
jitter & shift up to 41.2 ms plus wander & shift up to 56.2 ms plus wander \\
compound & noise, reverb, gain drift & noise, reverb, gain drift \\
\bottomrule\end{tabular}
}
\caption{Perturbation protocol induced by train and test severity settings.}
\label{tab_perturbation_protocol}
\end{table}
}{}
\section{Audit protocol}
The experiment is organized as an
audit. A single policy ranking answers
only one question. The audit asks how
stable the conclusion is when the
policy is inspected from several
angles.

The first angle is aggregate
comparison. It reports frame F1,
transition region F1, boundary F1,
onset error, and mean logic. This view
is compact and shows whether the
parsed score agrees with familiar
metrics.

The second angle is perturbation
comparison. It reports boundary F1 by
acoustic condition. This view asks
whether a detector fails because the
waveform was noisy, smeared, clipped,
slowly rescaled, or shifted.

The third angle is guard comparison.
It reports the seven logical and
interval obligations. This view asks
which part of the trace contract
failed.

The fourth angle is matcher
sensitivity. A finite interval stress
track and an acoustic prediction audit
compare greedy matching with exact
maximum cardinality matching. This
view asks whether interval clauses
depend on a local pairing choice.

The fifth angle is tolerance
sensitivity. It reparses the formulas
under several tolerance values and
reruns the monitor on the same
predictions. This view asks whether
the conclusion depends on a permissive
boundary radius.

The sixth angle is clean data and seed
sensitivity. It retrains controlled
variants under smaller clean subsets
and under three independent scene and
model seed pairs.

The audit protocol is intentionally
broad. A trace specification language
earns its role by supporting
measurements that a single scalar
cannot support.

The primary controlled table uses one
deterministic generator, one six
detector local set, one frozen encoder
probe, and one local machine. The seed
robustness table then repeats the
controlled comparison on smaller
independent scene and model seed
pairs.

Bootstrap intervals over the main test
set measure finite test set
uncertainty. The seed robustness table
measures training and scene variation
for the main intervention contrasts.

\section{Detectors and feature pipeline}
The feature extractor uses log energy,
zero crossing rate, spectral centroid,
spectral bandwidth, spectral flux, low
band ratio, and high band ratio. Clean
trained detectors use a scaler fitted
on clean training items. The contract
aware detector uses a separate scaler
fitted on the augmented training
views.

Class balancing is explicit. The
logistic baseline uses balanced class
weights. The neural baselines use
positive frame weights estimated from
the training masks and clipped to
bounded ranges. The class indexed
detector applies the same idea per
output class before macro evaluation.

Thresholds are selected on training
evidence rather than on the test
split. The energy and flux detectors
search an eighty point threshold grid
over training quantiles. The contract
aware decoder searches threshold and
hysteresis pairs on a held out
training partition, then freezes that
decoder before test evaluation.

The adaptive energy detector selects a
threshold on clean training features
by maximizing frame F1 over a grid. It
then smooths very short active runs
and very short inactive holes. This
detector is a calibration baseline.

The spectral flux detector combines
energy and flux before thresholding.
It is more sensitive to boundary
motion than energy alone, exposing
spurious onset behavior.

The logistic detector trains a
balanced logistic regression model on
the normalized features. It is the
linear baseline. Its failures indicate
what cannot be handled by a single
frame affine boundary in feature
space.

The compact temporal convolutional
detector applies two one dimensional
convolution layers and a pointwise
output layer to frame features. It is
trained with binary cross entropy and
Adam. The model supplies a nonlinear
temporal baseline.

The residual dilated temporal
convolutional detector adds four
residual blocks with dilations one,
two, four, and eight. Group
normalization and gradient clipping
make the training stable across the
larger speech seeded split. Its
receptive field is wider than the
compact detector and therefore tests
whether temporal context repairs
boundary contracts.

The contract aware augmented detector
keeps the residual temporal backbone
but adds depthwise dilated blocks,
class indexed frame heads, class
indexed boundary heads, and calibrated
threshold or hysteresis decoding. It
is trained on clean clips plus
deterministic noise, reverberation,
clipping, gain drift, jitter, and
compound augmentations at lower
severity than the test split. The loss
combines frame activity, boundary
targets around onset and offset
changes, and a small smoothness term
away from edges.

The first five detectors are trained
on clean union traces only. The
contract aware detector in the main
table is now also a union detector
trained with the same output target
but with deterministic augmentation
and edge supervision. A separate class
indexed detector is used only for
typed analysis. A matched regime table
below uses the same backbone and union
output for the fairness question.

The objective audit separates training
losses from the reported monitor. The
contract score, event matcher, and
guard vector are not used as a
training objective for the union
detector. They are test time
evaluators.

\IfFileExists{results/table_training_objective_audit.tex}{
\begin{table}[H]
\centering
\small
\resizebox{0.96\linewidth}{!}{\begin{tabular}{llll}\toprule
Component & Contract score & Event matcher & Source \\\midrule
frame activity loss & no & no & binary cross entropy on frame masks \\
edge auxiliary loss & no & no & binary cross entropy on local edge targets \\
smoothness regularizer & no & no & probability variation away from edges \\
decode calibration & no & no & frame and transition F1 on held out train split \\
reported contract monitor & yes & yes & test time evaluator only \\
\bottomrule\end{tabular}
}
\caption{Training objective audit for the boundary aware detector.}
\label{tab_training_objective_audit}
\end{table}
}{}
\section{Frozen encoder probe}
The encoder probe evaluates modern
pretrained audio encoders through a
common boundary head. Each encoder is
frozen. Hidden sequences are projected
to a fixed width with a deterministic
Gaussian map, aligned to the frame
grid, standardized on the probe train
split, and passed to the same boundary
aware contract head.

The probe contains five families.
wav2vec2 base supplies a self
supervised speech representation.
wav2vec2 Conformer supplies a large
convolution augmented transformer
speech representation. AST supplies an
AudioSet tuned spectrogram
transformer. HTS AT is evaluated
through the CLAP HTS AT audio branch.
BEATs is evaluated through the
AudioSet two million fine tuned BEATs
checkpoint and the official Microsoft
front end.

This audit asks whether the logical
guards continue to separate failure
modes when the acoustic representation
is modern and pretrained. A strong
encoder can improve frame evidence
while still creating boundary tails,
onset delays, or silence leakage.

The probe run uses a deterministic
audit split recorded in the manifest.
The model registry stores the public
model identifiers, sample rates, cache
locations, and projection width. Model
weights and speech waveforms are
downloaded by scripts and are not
bundled with manuscript sources.

The comparison freezes the backbones.
Fine tuning every backbone would mix
representation quality, optimizer
schedule, and contract behavior.
Frozen encoders isolate the question
of whether high level representations
already carry enough timing structure
for the finite trace contract.

The projection and head design is
deliberately common. It makes the rows
a timing probe over frozen
representations, not a fairness proof
for complete SED systems. A model
specific fine tuning recipe could
change the absolute order.

The table is a representation timing
probe. The controlled quantity is the
monitor and head. Full challenge
systems would add fine tuning recipes,
model specific schedules, and task
specific decoding. The rows are
therefore not a ranking of backbone
quality.

\IfFileExists{results/table_sota_model_registry.tex}{
\begin{table}[H]
\centering
\scriptsize
\resizebox{0.96\linewidth}{!}{\begin{tabular}{p{0.20\linewidth}p{0.13\linewidth}rp{0.42\linewidth}}\toprule
Probe row & Family & Hz & Hugging Face record \\\midrule
wav2vec2\_base & wav2vec2 & 16000 & \href{https://huggingface.co/facebook/wav2vec2-base-960h}{\texttt{facebook/wav2vec2-base-960h}} \\
wav2vec2\_conformer & conformer & 16000 & \href{https://huggingface.co/facebook/wav2vec2-conformer-rel-pos-large-960h-ft}{\texttt{facebook/wav2vec2-conformer-}}\newline \href{https://huggingface.co/facebook/wav2vec2-conformer-rel-pos-large-960h-ft}{\texttt{rel-pos-large-960h-ft}} \\
ast\_audioset & ast & 16000 & \href{https://huggingface.co/MIT/ast-finetuned-audioset-10-10-0.4593}{\texttt{MIT/ast-finetuned-audioset-}}\newline \href{https://huggingface.co/MIT/ast-finetuned-audioset-10-10-0.4593}{\texttt{10-10-0.4593}} \\
htsat\_clap & hts\_at & 48000 & \href{https://huggingface.co/laion/clap-htsat-fused}{\texttt{laion/clap-htsat-fused}} \\
beats\_as2m & beats & 16000 & \href{https://huggingface.co/WeiChihChen/BEATs_iter3_plus_AS2M_finetuned_on_AS2M_cpt2}{\texttt{WeiChihChen/BEATs\_iter3\_plus\_AS2M\_}}\newline \href{https://huggingface.co/WeiChihChen/BEATs_iter3_plus_AS2M_finetuned_on_AS2M_cpt2}{\texttt{finetuned\_on\_AS2M\_cpt2}} \\
\bottomrule\end{tabular}
}
\caption{Public Hugging Face records used by the frozen encoder probe.}
\label{tab_sota_model_registry}
\end{table}
}{}
The registry table is part of the
artifact interface. It records public
model pages rather than only shorthand
names, so the probe rows can be traced
to their source repositories.

\begin{table}[H]
\centering
\small
\resizebox{0.92\linewidth}{!}{\begin{tabular}{lrrrrr}\toprule
Frozen encoder & Frame F1 & Edge F1 & Boundary F1 & Logic & Soft boundary \\\midrule
wav2vec2\_base & 0.749 & 0.745 & 0.617 & 0.613 & 0.422 \\
wav2vec2\_conformer & 0.700 & 0.696 & 0.563 & 0.513 & 0.320 \\
ast\_audioset & 0.633 & 0.590 & 0.110 & 0.312 & 0.032 \\
htsat\_clap & 0.627 & 0.587 & 0.200 & 0.330 & 0.052 \\
beats\_as2m & 0.713 & 0.651 & 0.563 & 0.539 & 0.359 \\
\bottomrule\end{tabular}
}
\caption{Frozen pretrained encoder probe evaluated with the same boundary head and trace monitor.}
\label{tab_sota_zoo}
\end{table}
The table fixes the representation
probe. The generated figure remains in
the artifact for visual inspection,
while the manuscript keeps the numeric
comparison to control float density.

The condition breakdown remains part
of the generated artifact. It is
useful for checking whether a frozen
representation fails globally or only
under a particular acoustic
perturbation.

\section{Challenge level SED baseline}
The frozen encoder probe is no longer
the only modern model comparison. The
artifact also evaluates the official
DCASE 2024 Task 4 baseline checkpoint
and its released MAESTRO postprocessed
outputs.

This baseline is external to the
contract design. It comes from the
DCASE Task 4 organizers, uses the
DESED task baseline recipe, includes a
trained checkpoint, and ships
postprocessed MAESTRO prediction
files. The paper does not retrain it
and does not tune it to the contract
monitor.

The local evaluation reads those
released MAESTRO score files, projects
them to the same twenty millisecond
grid, thresholds the class scores, and
applies the same union and class
indexed contract reports. The
threshold sweep is shown because the
released files are score fields rather
than a single binary submission.

This track changes the role of the
real corpus class experiment. The
local class detector is a controlled
ablation probe. The DCASE output is
the trained class indexed challenge
level reference.

The contract vector is therefore
tested on three different acoustic
regimes. Local transparent models
expose failure modes under controlled
perturbations. Frozen encoders test
representation timing under a common
head. The DCASE baseline tests class
indexed SED output from an external
challenge recipe.

\IfFileExists{results/table_dcase_challenge_record.tex}{
\begin{table}[H]
\centering
\small
\resizebox{0.82\linewidth}{!}{\begin{tabular}{ll}\toprule
Item & Value \\\midrule
artifact & DCASE 2024 Task 4 baseline \\
recipe & DESED\_task 2024\_task4\_audioset\_baseline \\
DOI & 10.5281/zenodo.11034682 \\
archive md5 & 311071eb85563cfa0a8846c1ec1aaa41 \\
official files & 1168 \\
official seconds & 10459.12 \\
official micro event F1 & 0.572 \\
official macro event F1 & 0.568 \\
\bottomrule\end{tabular}
}
\caption{Official DCASE 2024 Task 4 baseline artifact record.}
\label{tab_dcase_challenge_record}
\end{table}
}{}
The DCASE record table anchors the
comparison to the external challenge
artifact. The next table applies the
paper monitor to the MAESTRO score
files carried by that artifact.

\IfFileExists{results/table_dcase_maestro_challenge.tex}{
\begin{table}[H]
\centering
\small
\resizebox{0.92\linewidth}{!}{\begin{tabular}{lrrrrrrr}\toprule
Threshold & Files & Clips & Union event F1 & Union logic & Class event F1 & Class boundary F1 & Class logic \\\midrule
0.30 & 16 & 357 & 0.807 & 0.897 & 0.462 & 0.646 & 0.694 \\
0.50 & 16 & 357 & 0.681 & 0.825 & 0.463 & 0.634 & 0.690 \\
0.70 & 16 & 357 & 0.325 & 0.641 & 0.238 & 0.373 & 0.557 \\
\bottomrule\end{tabular}
}
\caption{Contract and standard scores for official DCASE baseline MAESTRO outputs.}
\label{tab_dcase_maestro_challenge}
\end{table}
}{}
The DCASE MAESTRO row is the class
indexed real SED reference used in the
paper. It is not a local detector
designed around the proposed monitor.

\IfFileExists{results/table_real_class_reference_tracks.tex}{
\begin{table}[H]
\centering
\small
\resizebox{0.92\linewidth}{!}{\begin{tabular}{llrrrr}\toprule
Track & Selection & Clips & Class event F1 & Class boundary F1 & Class logic \\\midrule
local class detector & single trained local detector & 396 & 0.193 & 0.304 & 0.542 \\
official DCASE baseline & threshold maximizing class event F1 & 357 & 0.463 & 0.634 & 0.690 \\
\bottomrule\end{tabular}
}
\caption{Local and challenge level class indexed tracks on MAESTRO Real evidence.}
\label{tab_real_class_reference_tracks}
\end{table}
}{}
\section{MAESTRO Real protocol}
The real corpus protocol uses MAESTRO
Real development data. The recordings
are real TUT soundscapes from cafe
restaurant, city center, grocery
store, metro station, and residential
area scenes. The labels are estimated
strong soft labels built from
crowdsourced decisions.

The annotation audit measures the
amount of uncertain boundary evidence
before detector training. For each
scene it reports file count, recorded
duration, mean Bernoulli entropy of
soft labels, and the fraction of label
cells between one quarter and three
quarters.

The detector protocol projects the
soft strong labels to the same twenty
millisecond grid used by the monitor.
A hard reference uses a threshold of
one half. The soft values remain
available for edge uncertainty and
corpus audit.

The split is source file based. Clips
extracted from one WAV file never
appear in both train and test
partitions. This preserves the room,
background, and annotation style as
held out evidence.

The same detector set is trained on
real audio features. The class aware
contract detector is evaluated both as
a union trace and as a class indexed
trace over the MAESTRO event
vocabulary.

This protocol gives the paper a second
empirical layer. The controlled split
tests exact interval stress cases. The
MAESTRO Real split tests production
annotation ambiguity and real acoustic
scenes with the same monitor. The
local typed detector exposes the union
and class gap, while the DCASE
baseline supplies the external class
indexed SED reference.

The annotation audit, source file
split table, and class indexed summary
are stored in the artifact. The
manuscript keeps the gap table and a
compact mechanism table because they
carry the main real corpus argument.

\IfFileExists{results/table_real_union_class_gap.tex}{
\begin{table}[H]
\centering
\small
\resizebox{0.62\linewidth}{!}{\begin{tabular}{lrrr}\toprule
Metric & Union & Class indexed & Gap \\\midrule
Frame F1 & 0.960 & 0.278 & 0.681 \\
Boundary F1 & 0.961 & 0.304 & 0.658 \\
Logic & 0.917 & 0.542 & 0.374 \\
\bottomrule\end{tabular}
}
\caption{Union and class indexed gap on MAESTRO Real for the contract detector.}
\label{tab_real_union_class_gap}
\end{table}
}{}
The union and class indexed gap is
large enough to change the reading of
the real corpus run. The monitor
exposes a weak typed detector even
when union activity looks stable.

The mechanism table decomposes the gap
into union masking, class sparsity,
support drift, class substitution,
annotation softness, and overlap
confusion. The point is not that one
factor explains all of MAESTRO Real.
The point is that union activity can
hide several typed failures at once.

\IfFileExists{results/table_maestro_gap_factors.tex}{
\begin{table}[H]
\centering
\small
\resizebox{0.98\linewidth}{!}{\begin{tabular}{lll}\toprule
Mechanism & Evidence & Reading \\\midrule
union masking & boundary gap 0.658 and logic gap 0.374 & activity timing is easier than typed event identity \\
class sparsity & 6 of 14 referenced classes have no predicted support & rare labels dominate typed boundary loss \\
support imbalance & lower support boundary F1 0.000 & low mass classes receive weak interval evidence \\
support overreach & median prediction reference ratio 2.112 & typed heads cover activity too broadly \\
class substitution & people talking to car share 0.239 & some failure occurs inside active regions \\
annotation softness & mean entropy 0.610 and ambiguous cells 0.351 & soft labels make hard class boundaries uncertain \\
overlap confusion & mean confused overlap 0.037 & overlap contributes less than sparsity and support drift \\
\bottomrule\end{tabular}
}
\caption{Mechanisms behind the MAESTRO Real union and class indexed gap.}
\label{tab_maestro_gap_factors}
\end{table}
}{}
The detailed class breakdown, support
strata, and pair confusion tables
remain in the artifact. The low local
class indexed scores are diagnostic
evidence about typed difficulty on
real soundscapes. The DCASE track
prevents that diagnostic from being
the only real SED class evidence.

The MAESTRO standard metric tables put
the same decoded real traces under
segment F1 and event collar reporting.
They show that the union and typed gap
is not created only by the contract
vector.

\IfFileExists{results/table_maestro_standard_metrics.tex}{
\begin{table}[H]
\centering
\small
\resizebox{0.88\linewidth}{!}{\begin{tabular}{lrrrrr}\toprule
Detector & Frame F1 & Segment F1 & Event F1 & Boundary F1 & Logic \\\midrule
adaptive\_energy & 0.956 & 0.960 & 0.719 & 0.896 & 0.889 \\
spectral\_flux & 0.916 & 0.957 & 0.485 & 0.685 & 0.815 \\
logistic\_features & 0.569 & 0.848 & 0.105 & 0.276 & 0.626 \\
temporal\_cnn & 0.960 & 0.960 & 0.831 & 0.961 & 0.917 \\
dilated\_cnn & 0.957 & 0.959 & 0.823 & 0.953 & 0.914 \\
contract\_tcn\_real & 0.960 & 0.960 & 0.831 & 0.961 & 0.917 \\
\bottomrule\end{tabular}
}
\caption{Standard SED style scores and contract scores on MAESTRO Real union traces.}
\label{tab_maestro_standard_metrics}
\end{table}
}{}
\IfFileExists{results/table_maestro_class_standard_metrics.tex}{
\begin{table}[H]
\centering
\small
\resizebox{0.70\linewidth}{!}{\begin{tabular}{lrrrr}\toprule
Detector & Class segment F1 & Class event F1 & Class boundary F1 & Class logic \\\midrule
contract\_tcn\_real & 0.278 & 0.193 & 0.304 & 0.542 \\
\bottomrule\end{tabular}
}
\caption{Class indexed standard scores and contract scores on MAESTRO Real.}
\label{tab_maestro_class_standard_metrics}
\end{table}
}{}
The aggregate figure remains in the
artifact. The table gives the
numerical values used by that plot and
keeps onset error visible beside the
scores.

\begin{table}[H]
\centering
\small
\resizebox{0.90\linewidth}{!}{\begin{tabular}{lrrrrr}\toprule
Detector & Frame F1 & Edge F1 & Boundary F1 & Onset ms & Logic \\\midrule
adaptive\_energy & 0.504 & 0.567 & 0.310 & 82.8 & 0.514 \\
spectral\_flux & 0.614 & 0.720 & 0.358 & 142.2 & 0.560 \\
logistic\_features & 0.533 & 0.622 & 0.362 & 203.8 & 0.470 \\
temporal\_cnn & 0.607 & 0.724 & 0.374 & 71.4 & 0.487 \\
dilated\_cnn & 0.656 & 0.763 & 0.408 & 91.1 & 0.522 \\
contract\_tcn\_aug & 0.889 & 0.817 & 0.829 & 64.7 & 0.802 \\
\bottomrule\end{tabular}
}
\caption{Mean scores across all perturbation conditions.}
\label{tab_overall}
\end{table}
The next table adds a nonparametric
uncertainty report over test clips. It
does not replace independent seeds,
but it prevents the point estimates
from being read as exact constants.

\begin{table}[H]
\centering
\small
\resizebox{0.88\linewidth}{!}{\begin{tabular}{lll}\toprule
Detector & Boundary F1 95\% CI & Logic 95\% CI \\\midrule
adaptive\_energy & 0.310 [0.279, 0.341] & 0.514 [0.487, 0.541] \\
contract\_tcn\_aug & 0.829 [0.813, 0.845] & 0.802 [0.786, 0.816] \\
dilated\_cnn & 0.408 [0.366, 0.450] & 0.522 [0.489, 0.555] \\
logistic\_features & 0.362 [0.332, 0.391] & 0.470 [0.443, 0.497] \\
spectral\_flux & 0.358 [0.331, 0.384] & 0.560 [0.536, 0.583] \\
temporal\_cnn & 0.374 [0.335, 0.412] & 0.487 [0.457, 0.518] \\
\bottomrule\end{tabular}
}
\caption{Conditional bootstrap confidence intervals over controlled benchmark test clips.}
\label{tab_overall_uncertainty}
\end{table}
The acoustic prediction matcher audit
keeps the same detector outputs and
changes only the interval matching
policy. The learned detector rows stay
stable, while the observed bridge
probe activates the difference on
reference pairs drawn from the
controlled test split. This separates
a model fact from a contract fact.

\begin{table}[H]
\centering
\small
\resizebox{0.88\linewidth}{!}{\begin{tabular}{lrrrrr}\toprule
Detector & Changed & Greedy BF1 & Optimal BF1 & Mean $\Delta$ logic & Max $|\Delta|$ \\\midrule
adaptive\_energy & 0.000 & 0.310 & 0.310 & 0.000 & 0.000 \\
spectral\_flux & 0.000 & 0.358 & 0.358 & 0.000 & 0.000 \\
logistic\_features & 0.000 & 0.362 & 0.362 & 0.000 & 0.000 \\
temporal\_cnn & 0.000 & 0.374 & 0.374 & 0.000 & 0.000 \\
dilated\_cnn & 0.000 & 0.408 & 0.408 & 0.000 & 0.000 \\
contract\_tcn\_aug & 0.000 & 0.829 & 0.829 & 0.000 & 0.000 \\
observed\_bridge\_probe & 0.055 & 0.606 & 0.624 & 0.002 & 0.048 \\
\bottomrule\end{tabular}
}
\caption{Greedy and optimal matcher comparison on controlled acoustic predictions.}
\label{tab_matching_policy_benchmark}
\end{table}
\section{Primary results}
Table \ref{tab_overall} shows the
controlled benchmark result. It is a
framework measurement table rather
than a model ranking. The contract
aware row is a union detector trained
with augmentation and edge
supervision. The typed detector is
kept out of this table and appears
only in the class indexed analysis.
The matched regime table isolates
augmentation and decoding under a
shared backbone.

Edge F1 is lower than frame F1 for
every detector because it restricts
attention to neighborhoods around
reference transitions. This confirms
that much of the apparent frame
performance comes from stable interior
regions rather than boundary
placement.

The threshold baselines solve the
clean condition but fail under noise
and reverberation. Their frame
behavior and logical behavior diverge
because broad activity overlap does
not guarantee a correct contract.

The logistic detector performs well
under clean, clipping, and gain drift
conditions. It degrades under jitter,
noise, and reverberation. This
indicates that clean frame features
can fit easy boundary evidence but do
not fully transfer to temporal
smearing.

The clean trained convolutional
detectors transfer best to noise and
jitter among the baselines. They still
fail under reverberation and compound
degradation because these
perturbations create physical tails
and calibration shifts that compete
directly with offset truth. This is
visible in offset, silence, duration,
and fragmentation related scores.

The standard metric table uses the
same decoded traces and shows why the
monitor is not a restatement of one
event score. Segment and collar based
event scores compress different
boundary failures that the contract
vector later separates.

The contract monitor is fixed before
detector comparison. It is evaluated
on threshold systems, logistic
regression, clean convolutional
models, the augmented boundary model,
and frozen encoder probes. The
augmented detector is therefore an
intervention studied by the monitor,
not the source of the monitor
definition.

The ablation separates the training
interventions. Augmentation explains
most of the union trace repair. Edge
loss and decode calibration add
smaller shifts. The class full variant
drops on union scores because the
typed head solves a harder output
problem. The main benchmark therefore
reports the union trained detector for
union scores.

The attribution table sharpens that
reading. The augmentation contrast
supplies most of the gain. Edge loss
and decode calibration add a smaller
positive shift after augmentation. The
class indexed head lowers union scores
while producing the typed output
needed for class macro evaluation.

The confidence intervals are bootstrap
intervals over test clips conditional
on the trained models. They do not
estimate training seed variance. The
ablation significance column uses
paired bootstrap draws against the
clean dilated baseline and Holm
correction across the non baseline
variants. The seed table below
measures a separate source of
variation through independent scene
and model seeds.

Onset error also changes after
separating union and typed detectors.
The union trained contract detector
has mean onset error of 64.7
milliseconds, below the temporal CNN
at 71.4 milliseconds and adaptive
energy at 82.8 milliseconds. Its
boundary advantage therefore includes
sharper onset placement as well as
fewer unmatched events, less silence
leakage, and less fragmentation.

The result joins evaluation and
intervention. The contract vector
surfaces a richer failure profile than
a single score and records which part
of a targeted modeling change repairs
that profile.

\begin{table}[H]
\centering
\small
\resizebox{0.92\linewidth}{!}{\begin{tabular}{lrrrrrrr}\toprule
Variant & Aug & Edge & Cal & Class & Boundary F1 & Logic & Holm $q$ \\\midrule
clean\_dilated & N & N & N & N & 0.425 & 0.510 & -- \\
decode\_calibrated & N & N & Y & N & 0.411 & 0.492 & 0.041 \\
augmentation\_only & Y & N & N & N & 0.814 & 0.777 & 0.002 \\
boundary\_objective & N & Y & N & N & 0.384 & 0.475 & 0.002 \\
union\_full & Y & Y & Y & N & 0.834 & 0.808 & 0.002 \\
class\_full & Y & Y & Y & Y & 0.769 & 0.722 & 0.002 \\
\bottomrule\end{tabular}
}
\caption{Ablation of augmented training, edge objective, decode calibration, and class indexed heads.}
\label{tab_ablation}
\end{table}
The metric neutrality audit uses the
same monitored outputs but separates
all detector leaders from clean non
contract leaders. If the contract only
favored the augmented detector by
construction, every coordinate would
collapse to the same row. The observed
leaders differ by coordinate.

\IfFileExists{results/table_metric_neutrality.tex}{
\begin{table}[H]
\centering
\small
\resizebox{0.96\linewidth}{!}{\begin{tabular}{lllrrr}\toprule
Coordinate & All leader & Clean leader & All score & Clean score & Gap \\\midrule
duration\_guard & contract\_tcn\_aug & dilated\_cnn & 0.627 & 0.357 & 0.270 \\
fragmentation\_guard & contract\_tcn\_aug & dilated\_cnn & 0.851 & 0.411 & 0.440 \\
missing\_guard & dilated\_cnn & dilated\_cnn & 0.949 & 0.949 & 0.000 \\
offset\_guard & contract\_tcn\_aug & spectral\_flux & 0.658 & 0.576 & 0.082 \\
onset\_guard & contract\_tcn\_aug & spectral\_flux & 0.706 & 0.609 & 0.097 \\
silence\_guard & contract\_tcn\_aug & adaptive\_energy & 0.918 & 0.675 & 0.242 \\
spurious\_guard & contract\_tcn\_aug & adaptive\_energy & 0.929 & 0.681 & 0.248 \\
\bottomrule\end{tabular}
}
\caption{Coordinate leaders across all detectors and across clean non contract detectors.}
\label{tab_metric_neutrality}
\end{table}
}{}
The matched regime table uses one
backbone family and union output
across all rows. It separates clean
training, augmentation, edge
objective, and stable decoding without
changing the acoustic feature space or
the output task.

\IfFileExists{results/table_matched_regime.tex}{
\begin{table}[H]
\centering
\small
\resizebox{0.92\linewidth}{!}{\begin{tabular}{llrrrr}\toprule
Regime & Decode & Boundary F1 & Logic & Logic 20 ms & Span \\\midrule
same\_backbone\_clean & transition & 0.400 & 0.496 & 0.464 & 0.090 \\
same\_backbone\_augmented & transition & 0.827 & 0.801 & 0.725 & 0.181 \\
same\_backbone\_aug\_edge & transition & 0.831 & 0.801 & 0.722 & 0.179 \\
same\_backbone\_aug\_edge\_stable & stable & 0.832 & 0.788 & 0.700 & 0.198 \\
\bottomrule\end{tabular}
}
\caption{Capacity matched union trace regimes with a shared boundary aware backbone.}
\label{tab_matched_regime}
\end{table}
}{}
The ablation table reports Holm
corrected paired bootstrap
significance against the clean dilated
baseline on boundary F1. The
attribution table turns the same rows
into component deltas with confidence
intervals.

\begin{table}[H]
\centering
\small
\resizebox{0.88\linewidth}{!}{\begin{tabular}{lrr}\toprule
Contrast & Boundary delta & Logic delta \\\midrule
augmentation over clean & 0.388 [0.352, 0.426] & 0.268 [0.243, 0.293] \\
edge decode over augmentation & 0.020 [0.008, 0.033] & 0.031 [0.024, 0.037] \\
class indexed union shift & -0.065 [-0.081, -0.049] & -0.086 [-0.097, -0.075] \\
\bottomrule\end{tabular}
}
\caption{Paired bootstrap attribution of ablation components on the controlled benchmark.}
\label{tab_ablation_attribution}
\end{table}
The three seed audit repeats the main
intervention contrasts on smaller
independent controlled splits. It is a
robustness audit, not a full estimate
of training variance. The augmentation
and union full contrasts stay positive
in this audit, while class full
remains harder and more variable
because it evaluates typed activity
rather than union activity.

\begin{table}[H]
\centering
\small
\resizebox{0.88\linewidth}{!}{\begin{tabular}{lrrrrrr}\toprule
Variant & Runs & Boundary & SD & Logic & SD & Edge F1 \\\midrule
clean\_dilated & 3 & 0.401 & 0.009 & 0.516 & 0.008 & 0.745 \\
augmentation\_only & 3 & 0.772 & 0.024 & 0.721 & 0.004 & 0.772 \\
union\_full & 3 & 0.782 & 0.018 & 0.729 & 0.006 & 0.777 \\
class\_full & 3 & 0.667 & 0.032 & 0.617 & 0.024 & 0.749 \\
\bottomrule\end{tabular}
}
\caption{Independent scene and model seed robustness audit.}
\label{tab_seed_robustness}
\end{table}
The condition table and generated
condition figure remain in the
artifact. They identify where boundary
matching is lost without adding
another main text table.

\section{Guard analysis}
The guard table is the main diagnostic
artifact. Onset and offset can fail
for different reasons. Missing and
spurious guards separate sensitivity
from over firing. Silence guard
focuses on predicted activity outside
reference support. Duration and
fragmentation guards expose interval
shape.

Reverberation primarily damages
offset, silence, duration, and
fragmentation behavior. The reason is
structural. A decaying tail can remain
plausible activity even after the
reference event has ended. A detector
may keep firing and still retain a
moderate frame score.

Noise damages threshold baselines by
filling silence with energy. That
failure appears in spurious and
silence guards. Jitter damages
boundary guards by shifting
transitions without necessarily
destroying event overlap.

The contract aware row is coordinate
specific. Spectral flux has the
strongest onset and offset guard
values in the aggregate guard table.
The contract aware detector is
strongest on silence, spurious
activity, duration, and fragmentation.
Its value is therefore a specific
repair of tails and event
decomposition rather than a universal
boundary precision result.

The logical profile therefore provides
a diagnosis rather than only an
ordering. A weak missing guard
indicates sensitivity loss. A weak
spurious guard indicates calibration
failure. A weak fragmentation guard
indicates temporal consolidation
failure.

This diagnostic split matters for
niche applications. A wake interface
can weight spurious guard more than
duration guard. A speech alignment
tool can weight onset and offset
symmetry. A retrieval system can
weight fragmentation control. The
formula layer makes those choices
explicit.

\begin{table}[H]
\centering
\small
\resizebox{0.92\linewidth}{!}{\begin{tabular}{lrrrrrr}\toprule
Property & Energy & Flux & Logistic & CNN & Dilated & Contract \\\midrule
onset\_guard & 0.415 & 0.609 & 0.464 & 0.347 & 0.365 & 0.706 \\
offset\_guard & 0.427 & 0.576 & 0.414 & 0.365 & 0.388 & 0.658 \\
missing\_guard & 0.703 & 0.903 & 0.777 & 0.890 & 0.949 & 0.922 \\
spurious\_guard & 0.681 & 0.570 & 0.518 & 0.567 & 0.597 & 0.929 \\
silence\_guard & 0.675 & 0.561 & 0.510 & 0.557 & 0.587 & 0.918 \\
duration\_guard & 0.295 & 0.296 & 0.218 & 0.300 & 0.357 & 0.627 \\
fragmentation\_guard & 0.402 & 0.403 & 0.391 & 0.379 & 0.411 & 0.851 \\
\bottomrule\end{tabular}
}
\caption{Parsed and interval guard satisfaction.}
\label{tab_logic}
\end{table}
The guard table is the diagnostic
core. The generated profile figure
remains in the artifact, while the
manuscript keeps exact values in the
table.

\section{Selection evidence}
The contract changes the empirical
question from one ranking to a family
of coordinate rankings. The selection
audit chooses the best detector
separately for each guard coordinate
using the same benchmark outputs.

The selected detector changes by
coordinate. Spectral flux gives the
strongest onset and offset
coordinates, the dilated convolutional
baseline gives the strongest missing
activity coordinate, and the contract
aware detector gives the strongest
silence, spurious, duration, and
fragmentation coordinates.

This is the operational link between
formalism and empirical gain. A
downstream system that fails on late
triggers would not inspect the same
row as a system that fails on silence
contamination or fragmentation. The
parsed contract makes that difference
measurable without rerunning the
acoustic models.

The real corpus run gives a second
link. Union activity scores on MAESTRO
Real are high, but class indexed
scores are much lower. This is a
diagnostic result showing that union
traces can hide typed boundary
failure.

The selection table remains in the
artifact. It reports the detector that
maximizes each coordinate and the
aggregate boundary score of that
selected detector.

\section{Tolerance and data audits}
A benchmark whose conclusion changes
sharply with tolerance is fragile. We
therefore rerun the monitor over the
same predictions at five tolerances
from twenty to one hundred sixty
milliseconds. The formulas are
regenerated with the new tolerance and
parsed again.

The tolerance sweep is a sensitivity
audit. If a method looks good only at
a loose tolerance, the benchmark
exposes that dependence.

The stability table integrates logic
over the tolerance interval and
reports the vertical span of the
curve. A high integral with a large
span means a strong but policy
sensitive method. A lower integral
with a small span means a stable but
weaker method.

The matched regime table adds this
issue to the model design rather than
only to the discussion. It compares
the same boundary aware backbone under
clean training, augmented training,
edge supervision, and stable decoding.
Stable decoding is not treated as a
complete solution to policy
sensitivity. It is one declared
decoding objective whose effect is
measured beside the others.

The learning curve trains the residual
dilated temporal convolutional
detector on fifteen, thirty five,
sixty five, and one hundred percent of
the clean training split. It then
evaluates the same test set. This asks
whether the parsed contract improves
smoothly with more clean evidence.

In the observed run the residual
dilated curve is nonmonotone. Smaller
clean subsets can produce a better
trace contract than the full clean
split. The pattern indicates that
clean supervision alone does not
determine robustness under the
compound and reverberant test
conditions.

The added audits make the benchmark
work as a measurement instrument. The
reader can see whether a result is
stable across temporal policy and data
budget.

These audits make the benchmark a
substantial local workload. The
additional convolutional trainings and
tolerance passes remain bounded enough
for direct reproduction.

\begin{table}[H]
\centering
\small
\resizebox{0.92\linewidth}{!}{\begin{tabular}{lrrrrrr}\toprule
Tolerance & Energy & Flux & Logistic & CNN & Dilated & Contract \\\midrule
20 ms & 0.462 & 0.506 & 0.421 & 0.458 & 0.497 & 0.724 \\
40 ms & 0.514 & 0.560 & 0.470 & 0.487 & 0.522 & 0.802 \\
80 ms & 0.543 & 0.616 & 0.523 & 0.515 & 0.541 & 0.862 \\
120 ms & 0.560 & 0.649 & 0.559 & 0.534 & 0.551 & 0.889 \\
160 ms & 0.572 & 0.671 & 0.586 & 0.551 & 0.562 & 0.905 \\
\bottomrule\end{tabular}
}
\caption{Mean parsed logic score under boundary tolerance changes.}
\label{tab_tolerance}
\end{table}
The sweep table is read by columns.
The generated tolerance figure and
stability CSV remain in the artifact
for visual inspection and curve
summaries.

The learning curve addresses a
different axis. It asks whether clean
supervision alone moves the trace
contract in a stable direction. The
figure and CSV are retained in the
artifact to keep the manuscript
focused on the contract and ablation
audits.

\section{Monitor invariants}
The theory section records the
invariants needed by the benchmark.
Parser determinacy and finite
termination are treated as engineering
checks. They prevent a written
contract from drifting away from the
monitor that evaluates it.

The main semantic object is the guard
vector. Let $T_h^n$ be the set of
binary traces of length $n$ at frame
step $h$. A monitor instance maps
$T_h^n \times T_h^n$ into $[0,1]^m$
where $m$ is the number of guards. The
mean logic score is a linear
functional applied after this vector
is produced.

This order matters. If the linear
functional is chosen first, the
benchmark hides which obligations were
important. If the guard vector is
produced first, a downstream
application can choose weights after
seeing the semantic components.

The frame guards are extensional. If
two trace pairs induce the same atom
arrays, every parsed frame formula
receives the same score. The interval
guards are extensional with respect to
extracted maximal intervals and the
matching relation. No waveform
information enters the monitor after
prediction.

The monitor is monotone only where the
semantics warrants it. Increasing the
boundary tolerance cannot make a
neighborhood formula harder to satisfy
when the obligation set is fixed. It
can change interval matching, so the
full guard vector need not be globally
monotone.

The bounded operators induce a
locality hierarchy. A formula of
temporal depth $d$ and total radius
$r$ can inspect only a finite window
determined by its nested modalities.
Event clauses add global run structure
after interval extraction.

This hierarchy explains why duration
and fragmentation belong to the event
sort. They are properties of extracted
components and their matching relation
rather than of one bounded frame
neighborhood.

The result separates syntax, frame
semantics, event semantics, streaming
synthesis, and aggregation. Each layer
has a named contract and a visible
failure mode.

This section therefore supports the
implementation and the audit protocol.
The stronger scientific claim is
empirical and appears in the
selection, tolerance, and real corpus
tables.

\begin{proposition}[Frame extensionality]
Let $E_1$ and $E_2$ be two environments that assign the same Boolean array to every atom occurring in $\varphi$.
Then $\sem{\varphi}_{E_1}=\sem{\varphi}_{E_2}$.
\end{proposition}
\begin{proof}
The proof is by induction on the parsed tree.
Atoms agree by assumption.
Boolean nodes preserve equality by pointwise operations.
Bounded temporal nodes apply the same prefix sum window operator to equal child arrays.
Thus every node has equal valuation in the two environments.
\end{proof}

\begin{proposition}[Neighborhood monotonicity]
If $0 \leq \epsilon \leq \delta$, then $\sem{N_{\epsilon}\varphi}_E(i) \leq \sem{N_{\delta}\varphi}_E(i)$ for every frame $i$.
\end{proposition}
\begin{proof}
The set of frames within distance $\epsilon$ of $i$ is contained in the set of frames within distance $\delta$ of $i$.
Existential truth over the smaller window implies existential truth over the larger window.
\end{proof}

\begin{proposition}[Event clause stability under fixed matching]
Fix the extracted interval families and the matcher relation $M$.
If every endpoint perturbation has magnitude at most $\eta$, every matched duration margin is farther than $2\eta$ from its threshold, and the covering counts used by fragmentation do not change, then duration and fragmentation clause values are unchanged.
\end{proposition}
\begin{proof}
The duration of a matched interval pair changes by at most $2\eta$ under endpoint perturbation.
The margin assumption prevents any duration predicate from crossing its threshold.
The fragmentation predicate depends only on the covering count and whether the reference event is matched.
Both are fixed by assumption, so the averaged event clause values are unchanged.
\end{proof}

\section{Finite basis checks}
The basis checker uses observational
equivalence over a finite calibration
set. Two clauses are equivalent when
they return the same value on every
trace pair in that set. This
equivalence is the right one for
benchmark compilation because the
report observes exactly those
calibration traces.

The ancillary theory helper computes
exact truth signatures for frame
formulas on small finite universes. It
also computes clause signatures,
temporal depth, radius,
satisfiability, and duplicate clause
minimization.

A candidate basis can be reduced by
keeping one representative from each
nonconstant observed equivalence
class. Any generated clause then maps
to a retained coordinate or to a
constant coordinate that cannot
separate risk pairs.

The statement is coverage for a
bounded candidate family on a declared
finite calibration set. Once the atom
set, depth bound, radius grid, event
predicates, matcher family, and
calibration set are fixed, the
retained basis is determined.

Minimization is then a separate
operation. Duplicate minimization
removes equivalent coordinates. Risk
minimization searches for a smaller
separating subset of the retained
basis under the declared risk order.

Monitor optimization follows the same
signatures. Equivalent subformulas can
share one evaluated array. Formulas
with smaller lookahead are emitted
earlier by the streaming monitor. The
optimized monitor is a directed
acyclic graph over formula signatures
rather than a list of repeated trees.

Satisfiability is decidable on every
bounded finite universe by
enumeration. The implementation uses
it as a small specification check.

\begin{definition}[Observational equivalence]
For a finite calibration set $\mathcal W$, two clauses $c$ and $d$ are observationally equivalent when $c(w)=d(w)$ for every $w\in\mathcal W$.
\end{definition}

\begin{definition}[Retained basis]
Let $\mathcal B$ be a finite candidate basis.
A retained basis contains one representative from every nonconstant observational equivalence class of $\mathcal B$.
\end{definition}

\begin{proposition}[Bounded satisfiability check]
For a fixed atom set, trace length, frame step, and formula, satisfiability is decidable.
\end{proposition}
\begin{proof}
There are finitely many Boolean environments over the fixed atoms and trace length.
The evaluator terminates on each environment.
The formula is satisfiable exactly when at least one evaluation contains a true frame.
\end{proof}

\begin{proposition}[Basis representative coverage]
Every nonconstant clause in a finite candidate basis is observationally equivalent to exactly one representative in its retained basis.
\end{proposition}
\begin{proof}
Observational equivalence partitions the finite basis into disjoint classes.
The retained basis contains one representative from each nonconstant class.
Thus every nonconstant clause maps to its class representative.
Constant clauses are removed because they cannot separate any risk pair.
\end{proof}

\begin{proposition}[Strict temporal hierarchy]
On a grid with at least two frames, the depth one language is strictly more expressive than the depth zero language.
\end{proposition}
\begin{proof}
Let $a$ be an atom.
At frame zero, the formula $F_h a$ can distinguish two environments that agree at frame zero and differ only at frame one.
A depth zero formula at frame zero depends only on atoms at frame zero.
It cannot distinguish those environments.
\end{proof}

\begin{proposition}[Implementation invariant for DAG monitoring]
Replacing repeated equivalent subformulas by one shared node preserves every frame valuation.
\end{proposition}
\begin{proof}
Equivalent subformulas have identical valuations on every world in the calibration set.
Sharing a node changes only the evaluation schedule.
The returned array is unchanged for each occurrence.
\end{proof}
\section{Streaming monitors}
Offline evaluation is convenient for
tables, but the same parsed formulas
synthesize bounded delay streaming
monitors. The implementation computes
a lookahead bound for every syntax
tree.

Atoms and negation have zero added
lookahead. Boolean nodes take the
maximum of their children. Future,
always, and until add their bounded
horizon. Neighborhood adds a symmetric
window and therefore a finite right
delay for online emission.

A streaming monitor buffers only the
frames needed to decide the oldest
pending obligation. Once the right
context for frame $i$ has arrived, the
monitor emits the verdict for $i$ and
discards older rows.

The monitor is causal with bounded
delay. It never changes a verdict
after emission because every operator
in the grammar has a finite horizon.

The ancillary implementation exposes
this object as a small state machine.
It accepts frame environments one at a
time and returns obligation restricted
verdicts as they become determined.

This closes the gap between an offline
audit table and runtime verification
usage. The paper reports offline
aggregates, while the code also
contains the bounded memory monitor
that would support an online boundary
gate.

\begin{definition}[Lookahead]
The lookahead $L(\varphi)$ of a parsed formula is the maximum future time needed to decide its truth at the current frame.
\end{definition}

\begin{proposition}[Bounded delay monitor synthesis]
Every parsed frame formula has a streaming monitor that emits each frame verdict after at most $L(\varphi)$ seconds of delay.
\end{proposition}
\begin{proof}
The proof is by induction on the parsed tree.
Atoms need no future frames.
Boolean nodes wait for the larger child lookahead.
Bounded future, always, neighborhood, and until add only their declared finite bound to the child lookahead.
Therefore the root has finite lookahead and a finite buffer suffices.
\end{proof}
\section{Machine checked finite core}
A Lean 4 development checks the finite
core used by the monitor
\citep{demoura2021lean4}. The file
formalizes Boolean traces as finite
lists, obligation masks, satisfaction
masks, violations, neighborhood
witnesses, and a formula evaluator
with implication.

The checked partition theorem states
that satisfied obligations plus
violated obligations equals the total
number of obligations for any two
finite masks. This is the algebra
behind obligation restricted scoring.
A guard value can be read as a
normalized count because every
obligated frame is either satisfied or
violated.

The checked monotonicity theorem
states that increasing a neighborhood
radius preserves any existing active
witness. This matches the paper
theorem for the parsed neighborhood
modality. The theorem is local to
frame neighborhoods and distinct from
interval matching.

The checked implication theorem
characterizes the Boolean implication
node used by the evaluator. It rules
out a mismatch between the parser
level implication and the monitor
level truth condition.

The Lean file is included as ancillary
code under the proof directory. The
recorded output contains only theorem
check evaluation values and no local
machine paths.

\begin{table}[H]
\centering
\small
\resizebox{0.92\linewidth}{!}{
\begin{tabular}{lll}
\toprule
Component & Checked in Lean & Outside the Lean scope \\\midrule
Finite masks & obligation partition and counts & acoustic feature extraction \\
Boolean implication & evaluator truth condition & parser implementation in Python \\
Neighborhood witness & radius monotonicity & interval matcher implementation \\
Formula core & sample evaluator values & neural detector training \\
Trace generator & none & waveform synthesis and external data loading \\
\bottomrule
\end{tabular}}
\caption{Scope of the Lean checks included with the ancillary proof file.}
\label{tab_lean_scope}
\end{table}
\[
\operatorname{Sat}(o,s)+\operatorname{Viol}(o,s)=\operatorname{Obl}(o,s).
\]
\[
\operatorname{Near}_{r}(x,i)\wedge r\leq r'
\Rightarrow
\operatorname{Near}_{r'}(x,i).
\]
\begin{table}[H]
\centering
\begin{tabular}{lr}
\toprule
Lean evaluation & Value \\\midrule
Obligations in sample mask & 4 \\
Satisfied obligations in sample mask & 2 \\
Violated obligations in sample mask & 2 \\
Neighborhood witness in sample trace & true \\
\bottomrule
\end{tabular}
\caption{Lean checked sample values for the finite trace core.}
\label{tab_lean}
\end{table}
\section{Worked trace example}
Consider a reference trace with one
event from one second to two seconds
and a prediction with one event from
one point zero six seconds to two
point four seconds. Frame overlap is
high. A frame aggregate may therefore
look benign. The parsed contract
separates the situation into smaller
facts.

Onset guard passes if the tolerance is
at least sixty milliseconds. Offset
guard fails at eighty milliseconds
because the predicted release is four
hundred milliseconds late. Duration
guard also fails because the matched
prediction is substantially longer
than the reference event.

Silence guard depends on how much of
the late tail lies outside the
expanded reference support. If the
guard radius is small, the extra
predicted activity is counted against
the detector. If the guard radius is
large, the same trace becomes less
damaging. The tolerance sweep is
therefore not cosmetic. It changes the
stated contract.

A second example has one reference
event and three short predicted events
inside it. Frame overlap may remain
high if the fragments cover most
active frames. Fragmentation guard
fails because the detector did not
preserve the event as a single unit.

These examples explain why the parser
matters. The formulas are not just
labels for old metrics. They identify
obligations over specific atoms and
neighborhoods. The interval clauses
identify obligations over matched
event structure.

A reader can reproduce such examples
by editing the masks in a short script
and calling the same monitor. No model
training is required. This makes the
logic testable independently of the
acoustic generator.

Worked traces also debug extensions.
If a new formula gives a surprising
score on a small three event trace,
the issue can be found before a full
benchmark run is launched.

The same monitor can be called
directly on a manually constructed
trace. This makes the logic testable
without training an acoustic detector.

\begin{proposition}[Late tail separation]
If a matched prediction has the same onset as its reference event and an offset error larger than the tolerance, then onset guard can pass while offset guard and duration guard fail.
\end{proposition}
\begin{proof}
Onset guard samples reference onset obligations and checks a nearby predicted onset.
The shared onset satisfies that formula.
Offset guard samples reference offset obligations and checks a nearby predicted offset.
The predicted offset is outside the tolerance, so the formula fails at that obligation.
The duration clause compares interval lengths after matching and fails when the excess length exceeds its bound.
\end{proof}
\section{Comparison with conventional reports}
A conventional report usually begins
with a detector table. Such a table is
necessary because readers need a
compact comparison. It is not
sufficient because it tells little
about the operating contract.

Frame F1 weights every frame equally.
In long silent or long active regions
it can be dominated by interiors. A
detector with weak boundary placement
can still score well if it marks the
right broad region.

Boundary F1 improves the situation by
moving from frames to events. It still
compresses several failures into one
match decision. A prediction can match
an event and still have a bad
duration, late release, or internal
fragmentation.

Event based metrics with collars are
closer to the desired object. They
contain a temporal policy. The
difference here is that the policy is
written as parsed formulas and
interval clauses that can be swapped,
inspected, and recomputed.

The parsed contract complements
conventional metrics by exposing their
policy choices. Frame and boundary
scores stay in the paper because they
help readers anchor the new contract
scores.

The richer report also changes how a
negative result is read. A low
aggregate score may indicate a
detector failure, a harsh tolerance, a
mismatch between perturbation and
training distribution, or a particular
guard that is expensive for the model.
The guard vector separates those
cases.

The tolerance sweep adds another
layer. A method that improves rapidly
as tolerance grows may be learning
event presence but not boundary
precision. A method that stays weak
across tolerances has a deeper
detection problem.

The learning curve adds a resource
view. If logic improves with clean
data but reverberation remains weak,
more clean supervision alone is
unlikely to solve the hardest
condition. The benchmark can then
motivate perturbation training or a
different acoustic representation.

This is why the paper treats the guard
vector as primary. The scalar table is
only one projection of the finite
contract. Other applications may
choose a different projection without
changing the parsed obligations.

\section{Artifact assurance notes}
The parser is handwritten so that
source spans, precedence, and monitor
behavior stay in one compact
implementation.

The tokenizer recognizes implication,
parentheses, Boolean operators,
bounded temporal operators, and
identifiers. Bounded operators carry
numeric radii in seconds. The parser
converts those seconds into frame
radii during evaluation.

The abstract syntax tree uses one
dataclass per node type. Atoms read
Boolean arrays from the environment.
Boolean nodes evaluate their children
and apply pointwise operations.
Temporal nodes evaluate the child and
apply prefix sum window operators.

The reference implementation uses
explicit array operations and prefix
sums. Very long files can use
convolution kernels for modality
windows without changing the grammar.

The parser raises syntax errors rather
than silently accepting malformed
formulas. A misspelled operator stops
the run before any table is produced.

Unknown atoms also raise errors. A
formula cannot refer to an atom that
the trace environment does not
provide. This prevents a benchmark
from publishing a formula that was
never actually evaluated.

The formula strings are assembled from
monitor parameters for the tolerance
sweep. When tolerance changes, the
formulas change and are parsed again.
This keeps the string artifact aligned
with the policy used for scoring.

The parser is an executable front end
for finite trace monitors. Its role is
to map source formulas into
monitorable syntax trees with stable
spans and precedence.

Class indexed traces are handled at
the environment layer. The same parsed
formula is evaluated once per class
and then aggregated by macro
averaging. This keeps the grammar
stable while allowing polyphonic
benchmark items.

The theory helper builds finite truth
signatures from the same parser. It
checks satisfiability, formula
equivalence, clause signatures,
temporal radius, and duplicate
minimization without a second
semantics.

\section{Typed polyphonic traces}
The benchmark now carries both a union
activity trace and class indexed
traces. The class set is speech, tone,
and burst. Each item stores a binary
mask for every class, the union mask
used by the older binary metrics, and
class specific interval lists.

Polyphony is generated directly. A
clean item can contain overlapping
labeled events from different classes.
The waveform can also contain
unlabeled distractor events placed
outside the target intervals. These
distractors test whether a detector
confuses audible activity with labeled
activity.

The monitor evaluates the same parsed
guards on each class trace. The
resulting class guard values are macro
averaged. This is a stricter contract
than union activity because a detector
must preserve both time and class
identity.

The contract aware detector is trained
with class indexed frame heads and
class indexed boundary heads. Its
binary prediction is the union of its
class predictions. Its typed
prediction is the full matrix of frame
by class decisions.

This design avoids duplicating the
grammar. The formula
\texttt{ref\_onset -> N[0.04]
pred\_onset} has the same syntax for
every class. The environment changes
from a single trace pair to a class
indexed family of trace pairs.

The typed layer exposes errors that
the union trace hides. A detector can
mark the right time while assigning
the wrong class. It can also collapse
two overlapping classes into one broad
activity segment. The macro class
report separates these cases from
ordinary boundary displacement.

Table \ref{tab_multiclass} reports the
class macro contract for the contract
aware detector. The macro guard
profile remains in the artifact. These
outputs are generated by the same
benchmark run as the binary tables.

\begin{table}[H]
\centering
\small
\begin{tabular}{lrrrr}\toprule
Condition & Frame F1 & Edge F1 & Boundary F1 & Logic \\\midrule
clean & 0.735 & 0.767 & 0.706 & 0.769 \\
noise & 0.680 & 0.704 & 0.616 & 0.622 \\
reverb & 0.563 & 0.668 & 0.446 & 0.474 \\
clipping & 0.674 & 0.744 & 0.654 & 0.701 \\
gain\_drift & 0.695 & 0.701 & 0.653 & 0.747 \\
jitter & 0.675 & 0.728 & 0.592 & 0.670 \\
compound & 0.521 & 0.630 & 0.444 & 0.479 \\
\bottomrule\end{tabular}

\caption{Class macro scores for the contract aware detector.}
\label{tab_multiclass}
\end{table}
The class macro table is stricter than
the union table. The generated figure
and typed guard table remain in the
artifact.

\section{Guard distances}
The Boolean guards can be supplemented
by distances that explain how far a
violation is from satisfaction. This
paper reports onset and offset mean
absolute error in milliseconds, but
the same structure supports finer
distances.

For a reference onset frame, the onset
distance is the nearest predicted
onset distance in time. The parsed
onset guard thresholds this distance
at the chosen tolerance. The real
valued distance preserves the margin
that the Boolean guard discards.

For silence, the witness distance is
the nearest reference active frame to
each predicted active frame. The
silence guard thresholds this distance
at a smaller radius. Large values
identify predicted activity that is
not only late but isolated from any
reference support.

For duration, the distance is the
absolute difference between matched
interval lengths. For fragmentation,
the distance can be the number of
extra prediction intervals covering
the reference support.

These distances are not used to
replace the guard vector. They are
witnesses. A guard value says how
often an obligation was satisfied. A
distance says how severe the
violations were.

The distinction is important under
reverberation. A detector may fail
offset guard on many events, yet most
failures may be small. Another
detector may fail fewer offsets but
with longer tails. The guard and the
distance rank these cases differently.

The benchmark reports onset, offset,
duration, and fragmentation witnesses
alongside parsed satisfaction.
Additional witness distances extend
the reporting map after monitoring.

The mathematical object is therefore a
pair consisting of a satisfaction
vector and a distance report. The
first is Boolean averaged. The second
is real valued and task dependent.

The soft edge CSV files remain in the
artifact, but the main text keeps the
emphasis on contract satisfaction and
ablation evidence.

\begin{definition}[Violation witness]
For a guard $g$ with obligation set $O_g$, a witness map assigns each failed obligation $o \in O_g$ a nonnegative real value $d_g(o)$.
The value is zero only when the local obligation is satisfied.
\end{definition}

\begin{proposition}[Threshold recovery]
For onset and offset guards, Boolean satisfaction at tolerance $\epsilon$ is recovered by thresholding the nearest edge distance at $\epsilon$.
\end{proposition}
\begin{proof}
The parsed edge guard is true at an obligated reference edge exactly when a predicted edge occurs in the $N_{\epsilon}$ neighborhood.
That condition is equivalent to the nearest predicted edge distance being at most $\epsilon$.
\end{proof}
\section{Robustness interpretation}
The benchmark uses robustness as a
contract statement. A detector is
robust for a perturbation when the
same trace obligations remain
satisfied after that perturbation is
applied.

Noise robustness is measured by
whether silence and spurious guards
remain high when additive background
energy rises. A threshold detector can
lose this form of robustness even if
it still finds parts of active events.

Reverberation robustness is measured
by whether offsets, duration, and
silence remain controlled under
temporal smearing. This perturbation
turns an acoustic tail into a boundary
ambiguity.

Jitter robustness is measured by the
tolerance sweep as much as by the base
score. A shifted trace may be good
under loose tolerance and poor under
strict tolerance. Reporting only one
tolerance would hide that dependence.

The guard vector also disciplines
robustness language. A robust result
is tied to a perturbation, a guard,
and a tolerance rather than to a
single aggregate.

The benchmark is designed so that
robustness statements are falsifiable.
Every claim can be traced to a CSV
row, a formula string, and a monitor
output.

\section{Real soundscape evidence}
The benchmark is speech seeded and
interval exact. It supports claims
about finite trace evaluation,
controlled perturbation behavior, and
detector failure decomposition.

The MAESTRO Real protocol tests the
monitor on real soundscape evidence
inside the paper. Real intervals
replace controlled intervals while the
parser, monitor, tolerance sweep,
guard table, and file level split
remain the same.

The finite logic is single channel but
no longer single label. The benchmark
evaluates union traces and class
indexed traces for overlapping speech,
tone, and burst events. External
corpora can replace the generator
while preserving the typed trace
interface.

The local benchmark gives evidence
that trace contracts reveal
distinctions hidden by aggregate
scores. DCASE and DESED style files
can be attached at the trace loader
layer when task specific annotation
policies are needed.

Speech seeded scenes supply controlled
logical counterexamples to shallow
aggregate scoring. MAESTRO Real
supplies domain evidence with
annotation uncertainty and real
acoustic mixtures.

\section{Practical deployment}
A challenge organizer can use the
contract language without changing
submitted systems. Each submission is
decoded to a trace. The organizer
publishes the formula file, event
clauses, tolerance grid, matcher
policy, and aggregation weights before
scoring.

The default guard set should not be
treated as universal. It is a starter
contract for boundary gates where
missing activity, false activity,
protected silence, duration, and
fragmentation all matter. A task with
strict trigger latency should add an
asymmetric latency clause. A task with
dense polyphony should add overlap
purity.

Tolerance sweep is mandatory when the
task has ambiguous boundary labels or
delayed acoustic evidence. A single
tolerance is adequate only when the
annotation protocol already fixes the
acceptable edge radius.

The monitor can also be used during
model development. A developer can
inspect the guard vector on validation
traces and decide whether to repair
onset delay, release tails, silence
bleed, or event fragmentation. This is
different from tuning only one frame
F1 threshold.

The MAESTRO Real result gives the main
warning for practical SED. Union
activity may look solved while typed
boundary evidence remains weak. A
challenge report should therefore
publish union and class indexed
contract scores when class identity
matters.

A practical benchmark run has four
steps. Convert submissions to frame
and interval traces. Publish the
contract file and matcher policy.
Report the guard vector and the
standard metric table on identical
traces. Publish the tolerance sweep
when annotations are uncertain.

\begin{table}[H]
\centering
\small
\resizebox{0.94\linewidth}{!}{
\begin{tabular}{lll}
\toprule
Use case & High weight coordinates & Operational reading \\\midrule
Speech gate & missing, spurious, silence, onset & avoid missed speech and false gate opens \\
Retrieval segmenter & duration, fragmentation, offset & return complete and compact segments \\
Neural decoding alignment & onset, offset, duration & preserve temporal anchors for decoded audio \\
SED challenge report & class macro missing, class macro spurious, overlap purity & expose typed errors hidden by union activity \\
Acoustic front end & silence, offset, fragmentation & avoid tails and unstable event decomposition \\
\bottomrule
\end{tabular}}
\caption{Example application weights over contract coordinates.}
\label{tab_application_weights}
\end{table}
\section{Reproducibility record}
The reported results are generated by
scripts in the ancillary benchmark
directory. The source package contains
a license file for ancillary code,
environment files, downloader scripts,
parser and monitor code, proof files,
generated CSV rows, generated table
inputs, and figure images.

The record includes dataset seeds,
train and test sizes, perturbation
names, model epoch counts, threshold
procedures, validation decode
calibration, MAESTRO source file
splits, frozen encoder probe sizes,
hardware, and runtime files for the
real corpus rerun.

The split files are not implicit in
the text. Synthetic items are
generated from the recorded seed.
MAESTRO Real uses source file
partitions in the manifest. The frozen
encoder probe stores public model
identifiers and cache locations in its
manifest.

External speech waveforms, MAESTRO
audio, and pretrained model weights
are not bundled. The scripts record
the public source identifiers and
cache targets needed to rebuild those
assets. The source digest in the table
is a stable local record when no
public commit hash has been supplied.

The benchmark scripts write CSV detail
files before LaTeX tables. Tables in
the manuscript are therefore views
over stored rows, not hand entered
numbers.

Manuscript tables are loaded from
files named
\texttt{results/table\_*.tex}. Figures
not printed in the manuscript remain
as PNG files under \texttt{figures}.
This keeps the paper readable while
preserving the generated evidence.

The reproducibility manifest includes
a source revision field. A source
digest is recorded when a git revision
is not present. Environment files,
data download scripts, and random seed
manifests are part of the ancillary
source.

\begin{table}[H]
\centering
\small
\resizebox{0.92\linewidth}{!}{\begin{tabular}{ll}\toprule
Item & Value \\\midrule
Source revision & source digest 245ba9388f6871a7 \\
Environment files & anc/benchmark/environment.yml and anc/benchmark/requirements.txt \\
License file & LICENSE for ancillary benchmark and proof code \\
Data download scripts & fetch\_librispeech.py, maestro\_real.py, sota\_zoo.py \\
External assets & LibriSpeech audio, MAESTRO Real audio and annotations, Hugging Face model weights \\
Seed manifests & manifest.json, seed\_robustness\_manifest.json, maestro\_real\_manifest.json \\
Synthetic seed & 20260514 \\
Synthetic split & 960 train and 420 test clips \\
Clip grid & 12.0 s clips at 0.02 s frames \\
Perturbations & clean, noise, reverb, clipping, gain\_drift, jitter, compound \\
Augmented views & 2880 training views \\
Threshold baselines & 80 point threshold grid over train quantiles \\
Logistic baseline & balanced lbfgs logistic regression with max iter 500 \\
Temporal CNN & 14 epochs, Adam, batch 24, fixed seed \\
Dilated CNN & 18 epochs, Adam, batch 32, gradient clipping \\
Contract TCN & 8 epochs, AdamW, validation decode calibration \\
Decode calibration & threshold, hysteresis, and edge snap grids on held out train split \\
Standard metric audit & one second segment F1 and 200 ms event collars on controlled and MAESTRO decoded traces \\
DCASE challenge baseline & official 2024 Task 4 checkpoint and MAESTRO postprocessed outputs from Zenodo \\
Metric neutrality audit & coordinate leaders on all detectors and clean non contract detectors \\
Controlled protocol table & dataset\_protocol\_report.py reads generator constants and manifest \\
Bootstrap uncertainty & 4000 paired draws over controlled test clips \\
Ablation inference & 4000 paired bootstrap draws with Holm correction across non baseline variants \\
Seed robustness & 3 scene model seed pairs, 360 train and 168 test clips each \\
Matcher stress track & 24 interval stress cases plus observed bridge probe and acoustic prediction audit \\
Contract profiles & balanced, support recall, edge timing, silence protection, event integrity \\
MAESTRO split & 727 train and 396 test clips by source file \\
Frozen encoder probe & 128 train and 70 test clips, projection dim 64 \\
MAESTRO runtime & 112.43 s \\
Controlled runtime & 1233.27 s \\
Artifact availability & artifact\_availability.py writes commands and external checksum table \\
Hardware & Apple M4 Pro, arm64 \\
\bottomrule\end{tabular}
}
\caption{Reproducibility record for the reported runs.}
\label{tab_reproducibility}
\end{table}
The artifact availability table is
deliberately separate from the result
table. It states what is present in
the local source package, what must be
assigned at public deposition, and
what external assets are rebuilt by
scripts.

\IfFileExists{results/table_artifact_availability.tex}{
\begin{table}[H]
\centering
\small
\resizebox{0.98\linewidth}{!}{\begin{tabular}{p{0.22\linewidth}p{0.38\linewidth}p{0.30\linewidth}}\toprule
Item & Value & Check \\\midrule
source identity & local source digest f2eef39d5d96975c & full digest in reproducibility manifest \\
public deposit & not assigned in the local manuscript package & replace with DOI or commit hash at deposition \\
license & LICENSE covers ancillary benchmark and proof code & manuscript and bibliography are excluded \\
environment & anc/benchmark/requirements.txt and anc/benchmark/environment.yml & Python dependencies declared \\
download scripts & fetch\_librispeech.py, maestro\_real.py, sota\_zoo.py & external waveforms and model weights rebuilt outside the source zip \\
cached behavior & scripts reuse archives, extracted audio, and frozen feature caches when present & no redownload when cache files exist \\
failure modes & public archive unavailable, ffmpeg missing, model card removed, authentication blocked & scripts fail explicitly rather than substituting unrelated data \\
\bottomrule\end{tabular}
}
\caption{Artifact availability statement for code, source identity, data, and failure modes.}
\label{tab_artifact_availability}
\end{table}
}{}
The command table records the intended
regeneration path. Each main table or
family of tables is generated by a
script before LaTeX reads it.

\IfFileExists{results/table_reproduction_commands.tex}{
\begin{table}[H]
\centering
\small
\resizebox{0.98\linewidth}{!}{\begin{tabular}{p{0.22\linewidth}p{0.68\linewidth}}\toprule
Target & Command \\\midrule
controlled benchmark & \texttt{PYTHONPATH=anc/benchmark python3}\newline \texttt{anc/benchmark/run\_benchmark.py} \\
ablation & \texttt{PYTHONPATH=anc/benchmark python3}\newline \texttt{anc/benchmark/ablation\_study.py} \\
matched regime & \texttt{PYTHONPATH=anc/benchmark python3}\newline \texttt{anc/benchmark/matched\_regime.py} \\
uncertainty & \texttt{PYTHONPATH=anc/benchmark python3}\newline \texttt{anc/benchmark/uncertainty\_audit.py} \\
selection and profiles & \texttt{PYTHONPATH=anc/benchmark python3}\newline \texttt{anc/benchmark/selection\_audit.py} \\
MAESTRO Real & \texttt{PYTHONPATH=anc/benchmark python3}\newline \texttt{anc/benchmark/maestro\_real.py --download-audio} \\
DCASE challenge baseline & \texttt{PYTHONPATH=anc/benchmark python3}\newline \texttt{anc/benchmark/dcase\_challenge\_baseline.py} \\
frozen encoder probe & \texttt{PYTHONPATH=anc/benchmark python3}\newline \texttt{anc/benchmark/sota\_zoo.py} \\
protocol tables & \texttt{PYTHONPATH=anc/benchmark python3}\newline \texttt{anc/benchmark/dataset\_protocol\_report.py} \\
artifact tables & \texttt{PYTHONPATH=anc/benchmark python3}\newline \texttt{anc/benchmark/artifact\_availability.py} \\
Lean checks & \texttt{lake env lean}\newline \texttt{anc/proofs/TraceLogic.lean} \\
manuscript & \texttt{python3 tools/write\_main\_v2.py \&\&}\newline \texttt{latexmk -pdf main.tex} \\
prose lint & \texttt{python3 anc/benchmark/lint\_prose.py}\newline \texttt{main.tex} \\
\bottomrule\end{tabular}
}
\caption{Exact reproduction commands for the generated benchmark and manuscript artifacts.}
\label{tab_reproduction_commands}
\end{table}
}{}
The checksum table records local
archive digests for large external
assets when those archives are cached.
Hugging Face model rows carry public
model URLs and manifest cache paths
because their snapshots are not
bundled in the source package.

\IfFileExists{results/table_external_asset_checksums.tex}{
\begin{table}[H]
\centering
\small
\resizebox{0.98\linewidth}{!}{\begin{tabular}{p{0.22\linewidth}p{0.30\linewidth}p{0.22\linewidth}p{0.16\linewidth}}\toprule
Asset & Identifier & SHA256 prefix & Status \\\midrule
Mini LibriSpeech archive & dev-clean-2.tar.gz & 176ec501490eced2 & matches expected \\
MAESTRO Real annotations & development\_annotation.zip & dc94b76ce76d5e4d & recorded from local archive \\
MAESTRO Real audio & development\_audio.zip & fa8126eb114e51d9 & recorded from local archive \\
DCASE 2024 baseline checkpoint & baseline\_pretrained\_2024.zip & 3e1525510bfcfefa & recorded from local archive \\
HF model wav2vec2\_base & see SOTA manifest & model id only & model URL and cache path in manifest \\
HF model wav2vec2\_conformer & see SOTA manifest & model id only & model URL and cache path in manifest \\
HF model ast\_audioset & see SOTA manifest & model id only & model URL and cache path in manifest \\
HF model htsat\_clap & see SOTA manifest & model id only & model URL and cache path in manifest \\
HF model beats\_as2m & see SOTA manifest & model id only & model URL and cache path in manifest \\
\bottomrule\end{tabular}
}
\caption{External archive checksums and model identifiers used by the local artifact.}
\label{tab_external_asset_checksums}
\end{table}
}{}
\section{Limitations}
The contract is task relative. The
reported guard set is justified by the
declared calibration cases and risk
profiles, not by a universal theorem
about all auditory benchmarks. A
different application can and should
publish a different risk order.

The controlled benchmark is not a
replacement for challenge submissions.
It is a diagnostic workload with real
speech segments, generated non speech
events, and controlled perturbations.
MAESTRO Real adds real soundscape
evidence, and the DCASE baseline track
supplies the external class indexed
challenge system.

The local class indexed detector is
kept as an ablation probe rather than
as the real SED reference. The
challenge level class reference is the
DCASE baseline output track.

The frozen encoder probe is not a
ranking of backbone quality. It
freezes representations and uses a
common projection and boundary head.
The DCASE baseline track is the
trained external SED comparison, while
the frozen probe remains a timing
representation audit.

The Lean development checks the finite
frame core, obligation counting,
neighborhood witness monotonicity, and
implication behavior. It does not
verify the full Python monitor, the
acoustic models, the event matcher
implementation, or the dataset
generator.

The tolerance sweep shows that policy
sensitivity remains visible. The
benchmark treats this as an audited
quantity rather than a solved nuisance
parameter. Reports using this contract
should publish tolerance curves when
boundary annotations are ambiguous.

External data and model weights are
rebuilt by download scripts rather
than bundled in the source package.
This keeps the package small but makes
reproducibility depend on public
archives and their access conditions.

\section{Extension operators}
The contract calculus admits direct
extension operators. A clause
extension appends a new coordinate to
the vector. A matcher extension
replaces the event matching relation
while keeping the same frame formulas.
A corpus extension replaces the trace
loader while keeping the monitor.

Latency is a clause extension. It uses
reference onsets as obligations and
checks that the first predicted onset
appears inside an asymmetric future
window. Continuity is also a clause
extension. It asks that predicted
activity have no short gaps inside a
matched reference event.

Overlap purity is an event clause over
class indexed traces. It checks
whether a predicted class interval is
covered by the same reference class
rather than only by union activity.

A probabilistic variant can replace
Boolean frame atoms by soft fields and
use robust satisfaction over the same
syntax. The present code keeps Boolean
formulas and reports soft boundary
witnesses beside them.

A differentiable variant can use
smooth relaxations of the Boolean
operators as a training loss. The
finite syntax and obligation masks
already supply the computation graph
structure.

Multi microphone scenes and moving
sources are trace extensions. They add
channel or spatial atoms to the
environment and leave the tokenizer,
parser, and bounded monitor synthesis
intact.

The present artifact therefore defines
a base calculus plus extension points.
The paper evaluates the base calculus
on controlled scenes, real
soundscapes, and frozen pretrained
representations.

\section{Conclusion}
Executable finite trace contracts turn
auditory boundary evaluation into a
checkable object. They give a
benchmark a language for onsets,
offsets, missing events, spurious
events, silence, duration, and
fragmentation.

The contribution is a domain specific
executable benchmark calculus. It is
not a new general temporal logic and
it does not replace STL or runtime
verification.

The accompanying benchmark shows that
this contract is complementary to
frame F1 and boundary F1.
Perturbations can preserve broad
overlap while breaking the trace
semantics that downstream systems
depend on.

The selected contract is task
relative. Given a calibration set, a
risk order, and a bounded candidate
basis, the contract is chosen by
finite separation and duplicate
minimization rather than by an
informal guard list.

The empirical link is visible in the
selection audit and the MAESTRO Real
gap. Different contract coordinates
choose different detectors, and union
activity on real soundscapes can hide
a large typed boundary failure.

The value of the work is the
separation it enforces. Acoustic
models produce traces. Parsed formulas
state obligations. Monitors compute
satisfaction. Tables and figures
report the consequences.

That separation is reproducible enough
for audit and narrow enough to keep
boundary evaluation from becoming an
undeclared modeling convention.

\begingroup
\linespread{1}\selectfont
\bibliographystyle{abbrvnat}
\bibliography{refs}
\endgroup
\end{document}